\documentclass[10pt,journal,compsoc]{IEEEtran}
\usepackage{tikz}
\usepackage{url}
\usepackage{pdfpages}
\usepackage{afterpage}
\usepackage{color, colortbl}
\usepackage{graphicx}
\usepackage{amsmath}
\usepackage{amsfonts}
\usepackage{import}
\usepackage{pifont}
\usepackage{multirow}
\usepackage{subcaption}
\usepackage{footmisc}
\usepackage{algorithm}
\usepackage[noend]{algpseudocode}
\usepackage{amssymb}
\usepackage{amsthm}
\usepackage{bm}
\newtheorem{definition}{Definition}[]

\newtheorem{theorem}{Theorem}[]
\newtheorem{lemma}{Lemma}[]
\usepackage[
separate-uncertainty = true,
multi-part-units = repeat
]{siunitx}

\def\BState{\State\hskip-\ALG@thistlm}

\definecolor{Gray}{gray}{0.9}
\definecolor{LightCyan}{rgb}{0.88,1,1}
\pdfoutput=1
\graphicspath {{figures/}}

\makeatletter
\DeclareRobustCommand*\textsubscript[1]{%
	\@textsubscript{\selectfont#1}}
\def\@textsubscript#1{%
	{\m@th\ensuremath{_{\mbox{\fontsize\sf@size\z@#1}}}}}
\makeatother

\usepackage[flushleft]{threeparttable}

\begin{document}
	
	\title{MG-WFBP: Merging Gradients Wisely for Efficient Communication in Distributed Deep Learning}
	
\author{Shaohuai Shi, \IEEEmembership{Member, IEEE,}
        Xiaowen Chu\IEEEauthorrefmark{1}\thanks{* Corresponding author.}, \IEEEmembership{Senior Member, IEEE,}
        and Bo Li, \IEEEmembership{Fellow, IEEE}
\IEEEcompsocitemizethanks{
	
	\IEEEcompsocthanksitem Shaohuai Shi and Bo Li are with the Department of Computer Science and Engineering, The Hong Kong University of Science and Technology, Kowloon, Hong Kong, China. 
	E-mail: \{shaohuais, bli\}@cse.ust.hk.
	
	\IEEEcompsocthanksitem Xiaowen Chu is with the Department of Computer Science, Hong Kong Baptist University, Kowloon, Hong Kong, China. 
	E-mail: chxw@comp.hkbu.edu.hk.}
}

\IEEEtitleabstractindextext{%
\begin{abstract}
Distributed synchronous stochastic gradient descent has been widely used to train deep neural networks (DNNs) on computer clusters. With the increase of computational power, network communications generally limit the system scalability. Wait-free backpropagation (WFBP) is a popular solution to overlap communications with computations during the training process. In this paper, we observe that many DNNs have a large number of layers with only a small amount of data to be communicated at each layer in distributed training, which could make WFBP inefficient. Based on the fact that merging some short communication tasks into a single one can reduce the overall communication time, we formulate an optimization problem to minimize the training time in pipelining communications and computations. We derive an optimal solution that can be solved efficiently without affecting the training performance. We then apply the solution to propose a distributed training algorithm named merged-gradient WFBP (MG-WFBP) and implement it in two platforms Caffe and PyTorch. Extensive experiments in three GPU clusters are conducted to verify the effectiveness of MG-WFBP. We further exploit trace-based simulations of 4 to 2048 GPUs to explore the potential scaling efficiency of MG-WFBP. Experimental results show that MG-WFBP achieves much better scaling performance than existing methods. 
\end{abstract}
}
\maketitle
\IEEEdisplaynontitleabstractindextext

\section{Introduction}
The data-parallel synchronous stochastic gradient descent (S-SGD) method is commonly used as the optimizer to train large-scale deep neural networks (DNNs) \cite{dean2012large}\cite{goyal2017accurate}. In S-SGD, the computing tasks for each mini-batch of training data are distributed to a cluster of computing nodes, and the individual results (e.g., gradients) are aggregated to update the global network model before the next iteration begins. However, with more computing nodes and the fast-growing computing power of hardware accelerators, the data communication between computing nodes gradually becomes the performance bottleneck \cite{watcharapichat2016ako}\cite{cui2016geeps}\cite{shi2018adag}\cite{wang2019impact}. For example, the computing power of Nvidia GPUs has increased by 30x in the last 10 years, whilst it took about 15 years for the network speed to improve from 10Gbps to 100Gbps. Hence it becomes a critical issue to address the imbalance between computing and communication.

Some recent works try to reduce the impact of data communication at either algorithmic or system level. On one hand, gradients could be quantized or sparsified \cite{alistarh2017qsgd}\cite{lin2018deep}\cite{wen2017terngrad}\cite{shi2019adistributed}\cite{shi2019ijcai} in order to reduce the amount of data to be exchanged so that the communication time could be reduced. But these methods usually sacrifice the training convergence speed. On the other hand, the high-performance computing (HPC) community has proposed several methods to improve the communication performance of the cluster by optimizing the hardware or communication software library \cite{potluri2013efficient}\cite{chen2019roundrobin}. In terms of hardware, InfiniBand (IB) and Omni-Path networks can provide much higher communication bandwidth and lower latency, and are deployed to shorten the performance gap between communication and computation \cite{bayatpour2017scalable}. Regarding the software, the implementation of message passing interface (MPI) has been optimized to support more efficient communication in DNN training \cite{bayatpour2017scalable}\cite{awan2017s}. Nvidia's NCCL\footnote{\url{https://developer.nvidia.com/nccl}} is another highly optimized communication library for deep learning (DL) frameworks on multi-GPU settings. The scaling efficiency of distributed DL systems can be modeled as a function of the communication-to-computation ratio \cite{wen2017terngrad}. For example, training ResNet-50 \cite{he2016deep} requires about 7.8 billion floating point operations in computation, while it needs to all-reduce 102 MB of data in one iteration. Higher communication-to-computation ratio results in lower scaling efficiency. 

The layered structure of DNNs makes it possible to overlap the communication and computation during the backward propagation \cite{awan2017s}\cite{zhang2017poseidon}\cite{shi2018performance}, which is known as wait-free backpropagation (WFBP). WFBP begins to exchange the gradients of a layer immediately after they have been calculated; so if the data communication time of a layer is shorter than the computation time of the gradients of its previous layer, then this communication cost can be fully hidden. However, if very fast hardware accelerators are used while the network speed is relatively slow (i.e., a high communication-to-computation ratio), there can exist many layers whose communication time is longer than the corresponding computation time. In such case, it becomes important to optimize the communications. We observe that the layer-wise gradient communication in WFBP is suboptimal due to the fact that all-reducing a small amount of data cannot fully utilize the network bandwidth in current network topology due to the startup time of message transmitting (or transmission latency). For example, on our 10GbE platform, all-reducing a set of 200 KB vectors across 8 nodes using MPI requires about 1.5 ms, while all-reducing a set of 400 KB vectors only requires 1.8 ms, which means that if we merge the two sets of 200 KB vectors to a single set of 400 KB vectors, then the total communication time can be reduced from 3 ms to 1.8 ms. The same phenomena can also be found in RDMA-based networks \cite{handley2017re}\cite{guo2016rdma}. You et al. \cite{you2017scaling} have also noticed this problem, and proposed a single-layer communication (SyncEASGD) method which merges the gradients of different layers into a single tensor and then transfers only once per iteration. As compared to the layer-wise communication in WFBP, it can eliminate most of the startup time of data communications. But in their proposed method, gradient communication can only start after the backward propagation, thus it misses the opportunity of overlapping the communication with computation. 

We argue that the best way to reduce the training time needs to consider not only how to overlap communication with computation, but also how to improve the communication efficiency by avoiding transmitting small messages. According to the taxonomy of efficient distributed DL~\cite{tang2020communication,shi2020quantitative}, our proposed method belongs to a kind of scheduling solution to improve the scalability of distributed training. 

In this paper, we first formulate the communication scheduling problem in S-SGD as an optimization problem that aims to minimize the total training time of an iteration. We then propose a merged-gradient wait-free backward propagation (MG-WFBP) method and prove its optimality. The time complexity of MG-WFBP is $O(L^2)$ where $L$ is the number of layers (or tensors) in the DNN, and it only needs to be executed once before the whole training process. We implement MG-WFBP atop the popular DL frameworks Caffe \cite{jia2014caffe} and PyTorch\footnote{\url{https://pytorch.org}} \cite{pytorch2019}, and make it publicly available\footnote{https://github.com/HKBU-HPML/MG-WFBP}. To validate the effectiveness of our proposed MG-WFBP, we evaluate its performance using various DNNs on multi-GPU settings with both 10Gbps Ethernet (10GbE) and 56Gbps InfiniBand (56GbIB) interconnects. On the relatively slow Nvidia Tesla K80 GPU clusters with 10GbE, MG-WFBP achieves about $1.2$x to $1.36$x improvement than the state-of-the-art communication algorithms WFBP and SyncEASGD, respectively. On the latest Nvidia Tesla V100 GPU clusters with 10GbE or 56GbIB, MG-WFBP achieves an average of 18.8\% faster than WFBP and SyncEASGD in terms of end-to-end training time. To investigate its performance on large clusters, we resolve to trace-based simulation (due to limited hardware resources) on 4-worker to 2048-worker clusters. In the 64-worker simulation, the results show that MG-WFBP achieves more than $1.7$x and $1.3$x speedups compared to WFBP and SyncEASGD respectively. This paper is an extension of our previous conference publication \cite{shi2019mgwfbp}, and we make the following new contributions.
\begin{itemize}
	\item We provide a complete proof of the optimality of MG-WFBP.
	\item We implement MG-WFBP on PyTorch and also make it open-source.
	\item We conduct extensive experiments on two Nvidia V100 GPU clusters with 10Gbps Ethernet and 56Gbps InfiniBand interconnects using six DNNs. 
	\item We verify that MG-WFBP is also robust to mixed precision training which is widely used in latest Nvidia GPUs and Google TPUs. 
\end{itemize}

The rest of the paper is organized as follows. We present the preliminaries in Section \ref{s:pre}, followed by the formulation of the existing problem in Section \ref{s:profor}. We derive an optimal solution to the problem and then present our MG-WFBP algorithm in Section \ref{s:method}. The system implementation atop PyTorch is present in Section \ref{s:system}. Section \ref{s:eval} demonstrates the experimental studies on the proposed method compared to existing methods. Section \ref{s:relatedwork} introduces the related work. We discuss some limitations and possible directions in Section~\ref{s:discission}, and finally we conclude this paper in Section \ref{s:conclusion}.

\section{Preliminaries}\label{s:pre}

For ease of presentation, we summarize the frequently used mathematical notations in Table \ref{table:notation}.
\begin{table}[!ht]
	\centering
	\caption{Frequently used notations}
	\label{table:notation}
	\begin{tabular}{|l|l|}
		\hline
		Name &  Description \\\cline{1-2}
		\hline
		\hline
		$N$ & The number of computing nodes in the cluster. \\
		$\alpha$ & Latency (startup time) of the network between two nodes. \\
		$\beta$ & Transmission time per byte between two nodes. \\
		$\gamma$ & Summation time of two floating point numbers in one node. \\
		$a$ & Latency (startup time) of all-reduce.\\
		$b$ & Transmission and computation time per byte of all-reduce. \\
		$M$ & The size of a message in bytes. \\\cline{1-2}
		$W$ & Weights of the DNN. \\		
		$D_i^g$ & The input data size for the $g^{th}$ node at the $i^{th}$ mini-batch.\\\cline{1-2}
		$L$ & The number of learnable layers (or tensors) of a DNN.\\
		$p^{(l)}$ & The number of parameters in the learnable layer $l$.\\
		$t_{iter}$ & Time of one training iteration with one batch  of data.\\
		$t_{f}$ & Time of the forward pass in each iteration.\\
		$t_{b}$ & Time of the backward propagation in each iteration.\\
		$t_{u}$ & Time of the model update in each iteration.\\
		$t_{b}^{(l)}$ & Time of the backward propagation of layer $l$ in each iteration.\\
		$\tau_{b}^{(l)}$ & The timestamp when layer $l$ begins to calculate gradients.\\
		$\tau_{c}^{(l)}$ & The timestamp when layer $l$ begins to communicate gradients.\\
		$t_{c}$ & Time of gradient aggregation in each iteration.\\
		$t_{c}^{(l)}$ & Time of gradient aggregation of layer $l$ in each iteration.\\
		$t_{c}^{no}$ & The non-overlapped communication cost in each iteration.\\
\hline
	\end{tabular}
\end{table}
\subsection{Mini-batch SGD}

Consider an $L$-layer DNN with a loss function $\mathcal{L}(W,D)$ which defines the difference between the prediction values and the ground truth over the training data set $D$, where $W$ is the set of model weights. To minimize the loss function, the mini-batch SGD updates the parameters iteratively. Typically, the $i^{th}$ iteration of the training includes four steps: 1) A mini-batch of data $D_i$ ($D_i\subset D$) is read as inputs of the DNN. 2) $D_i$ is fed forward across the neural network from layer $1$ to layer $L$ to compute the prediction values, and finally the loss function $\mathcal{L}(W,D)$ is computed. 3) The first order gradients w.r.t. parameters and inputs are calculated and backpropagated from layer $L$ to layer $1$. 4) Finally, the parameters are updated with the layer-wise gradients. The training is terminated when some stopping criteria are satisfied. The update of $W$ can be formulated as follows:
\begin{equation}
W_{i+1}=W_{i}-\eta\cdot\nabla\mathcal{L}(W_{i},D_{i}),
\end{equation}
where $\eta$ is the learning rate of SGD, $W_{i}$ denotes the weights at the $i^{th}$ iteration, and $\nabla\mathcal{L}(W_{i},D_{i})$ denotes the gradients. The time consumed in the training process is mainly in steps 2 and 3, because step 1 of the $i^{th}$ iteration can be scheduled to overlap with the $(i-1)^{th}$ iteration, and the time of step 4 is negligible. Therefore, we can simplify the timeline of SGD as a forward pass followed by a backward pass. The time of one iteration is represented by $t_{iter}=t_f+t_b$, where $t_f$ is the time of the forward pass, and $t_b$ is the time of the backward pass.

\begin{figure}[!ht]
	\centering
	\begin{subfigure}{0.48\textwidth}
		\includegraphics[width=\linewidth]{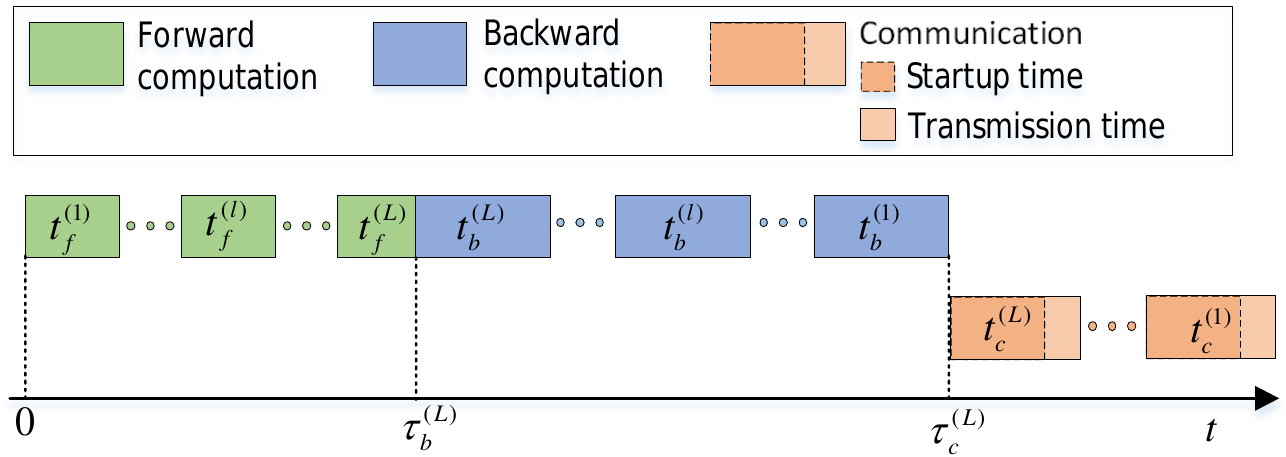}
		\caption{Naive S-SGD.}
	\end{subfigure}
	
	\begin{subfigure}{0.48\textwidth}
	   \vspace{10pt}
		\includegraphics[width=\linewidth]{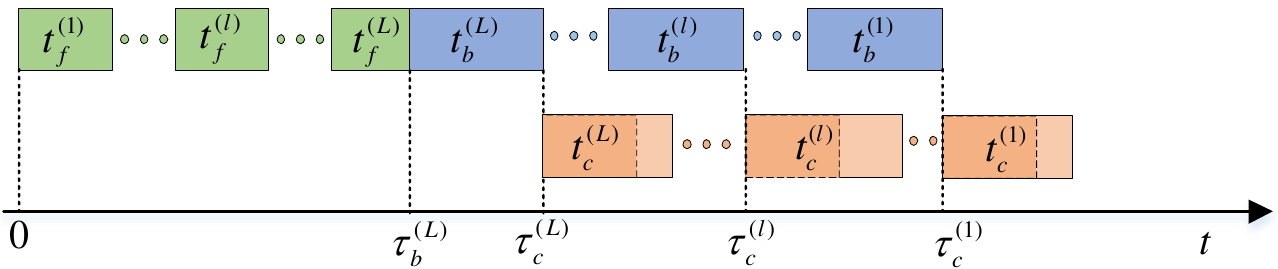}
		\caption{WFBP S-SGD.}
	\end{subfigure}
	
	\begin{subfigure}{0.48\textwidth}
	   \vspace{10pt}
		\includegraphics[width=\linewidth]{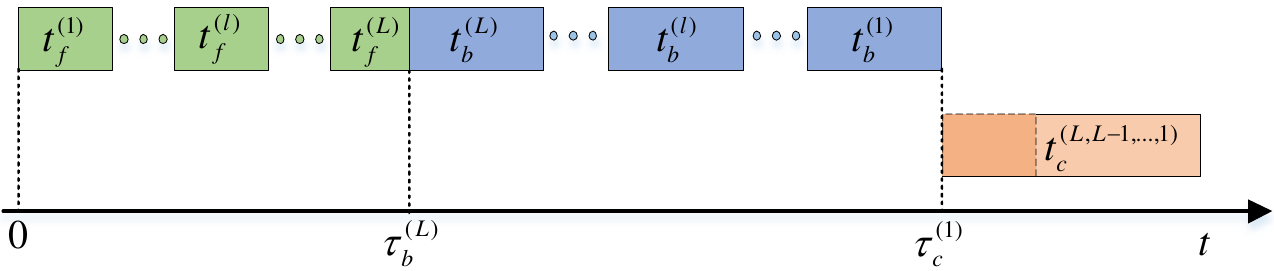}
		\caption{Single-layer S-SGD.}
	\end{subfigure}
	
	\caption{The timeline of the traditional S-SGD algorithms. (a) Naive S-SGD: Layer-wise gradient communications can only be started after all gradients have been calculated. (b) WFBP S-SGD (WFBP-SGD): Gradient communication of each layer begins immediately after the backward step of that layer. (c) SyncEASGD: All gradients are merged into a single-layer to be communicated together.}
	\label{fig:tranditionalSGDs}
\end{figure}

\subsection{Synchronized SGD}
For large-scale DNNs, the synchronized SGD (S-SGD) with data-parallelism is widely applied to train a model using multiple workers (say $N$ workers, and indexed by $g$). Each worker takes a different mini-batch of data $D_{i}^{g}$ and forwards it by step 2), and then follows step 3) to calculate the gradients $\nabla\mathcal{L}(W_{i},D_{i}^{g})$. In this way, each worker has a copy of the model, while the gradients calculated by different workers are different since the input data are different. At the end of each iteration of a mini-batch, S-SGD needs to average the gradients from different workers, updates the model by the averaged gradients, and synchronizes the model with all workers. The weights update formula of S-SGD is:
\begin{equation}\label{equ:ssgd}
W_{i+1}=W_{i}-\eta\cdot\frac{1}{N}\sum_{g=1}^{N}\nabla\mathcal{L}(W_{i},D_{i}^{g}).
\end{equation}
The averaging operation of gradients across the cluster involves extra computation and communication overheads. As a side-effect, it is not easy to achieve linear scaling in the distributed SGD training. The timeline of the naive S-SGD (i.e., computation and communication are not overlapped) with communication overheads is illustrated in Fig. \ref{fig:tranditionalSGDs}(a). The naive S-SGD algorithm suffers from the waiting period of data communication of model synchronization at every iteration. In practice, the gradients of a layer is stored as a tensor; hence the averaging process can be implemented by many all-reduce operations, once per layer. The layer-wise nature introduces many startup times for layer-wise gradients when they are communicated. The iteration time of the naive S-SGD can be estimated as
\begin{equation}\label{equ:tssgd}
t_{iter}=t_{f}+t_{b}+t_{c},
\end{equation}
where $t_{b}=\sum_{l=1}^{L}t_{b}^{(l)}$ is the layer-wise backward propagation time and $t_{c}=\sum_{l=1}^{L}t_{c}^{(l)}$ is the layer-wise gradient aggregation time which heavily relies on the communication performance.

Considering S-SGD running on $N$ workers, we define the speedup of S-SGD compared to the vanilla single-worker SGD:
\begin{equation}
S(N)=\frac{N|D_i^{g}|/(t_f+t_b+t_c)}{|D_i^{g}|/(t_f+t_b)}=\frac{N}{1+\frac{t_c}{t_f+t_b}},
\end{equation}
where $|D_i^{g}|$ is the number of training samples per worker at the $i^{th}$ iteration. Let $r=\frac{t_c}{t_f+t_b}$, which reflects the communication-to-computation ratio, we have
\begin{equation}\label{equ:speedupsync}
S(N)=\frac{N}{1+r}.
\end{equation}

\subsection{WFBP-SGD}
In WFBP S-SGD (WFBP-SGD), the gradient communication of layer $l$ ($l>1$) can be overlapped with the backward propagation of layer $l-1$. The timeline of WFBP-SGD is illustrated in Fig. \ref{fig:tranditionalSGDs}(b). For simplicity, we assume that the start timestamp of the forward pass is $0$, and the start timestamp of the backward pass is $\tau_b^{(L)} = t_f$. Then the timestamp when layer $l$ begins to calculate the gradients, denoted by $\tau_b^{(l)}$, can be calculated by:
\begin{equation}\label{equ:startcomp}
\tau_b^{(l)}=
\begin{cases}
t_f & l=L\\
\tau_b^{(l+1)}+t_b^{(l+1)} & 1\leq l<L
\end{cases}.
\end{equation}
Notice that the communication of gradients of layer $l$ ($l<L$) can only begin if the following two conditions are satisfied: (1) the gradients of layer $l$ have been calculated; (2) the communication of gradients of layer (l+1) has finished. So, the timestamp when layer $l$ begins the communication of gradients, denoted by $\tau_c^{(l)}$, can be calculated by:
\begin{equation}\label{equ:startt}
\tau_c^{(l)}=
\begin{cases}
\tau_b^{(l)}+t_b^{(l)} & l=L\\
\text{max}\{\tau_c^{(l+1)}+t_c^{(l+1)}, \tau_b^{(l)}+t_b^{(l)}\} & 1\leq l<L
\end{cases}.
\end{equation}
The iteration time of WFBP-SGD can be calculated as
\begin{equation}\label{equ:wfbpiter}
\begin{split}
t_{iter}&=t_f+t_b^{(L)}+(\tau_c^{(1)}-\tau_c^{(L)})+t_c^{(1)}\\
&=t_c^{(1)}+\text{max}\{\tau_c^{(2)}+t_c^{(2)}, \tau_b^{(1)}+t_b^{(1)}\}.
\end{split}
\end{equation}
Since some communications are overlapped with the computation, the non-overlapped communication cost, $t_{c}^{no}$, becomes the bottleneck of the system. In WFBP-SGD, we redefine $r=\frac{t_c^{no}}{t_f+t_b}$, so the main problem of WFBP-SGD is that when the communication cannot be fully overlapped by computation, i.e., $\tau_c^{(l+1)}+t_c^{(l+1)} >\tau_b^{(l)}+t_b^{(l)}$, $t_c^{no}$ will limit the system scalability.

\subsection{Single-Layer S-SGD}
As layer-wise communications introduce many startup times especially for large-scale clusters, the startup times dominate the communication time so that overlapping communications and computations may lead to even worse scaling efficiency. Therefore, You et al. \cite{you2017scaling} propose a single-layer communication mechanism (SyncEASGD) which merges all gradients to be communicated by a single all-reduce operation at the end of each iteration, as shown in Fig. \ref{fig:tranditionalSGDs}(c). The iteration time of SyncEASGD can be estimated as
\begin{equation}\label{equ:synceaiter}
t_{iter}=t_f+t_b+t_c,
\end{equation}
where $t_c$ is composed by the startup time and the transmission time.

\subsection{Communication Model}
In Eq. (\ref{equ:ssgd}), we use $\Delta W_i=\sum_{g=1}^{N}\nabla\mathcal{L}(W_{i},D_{i}^{g})$ to represent the aggregation of gradients from $N$ workers, which is an all-reduce operation\footnote{In this paper, we mainly discuss the scenario with the all-reduce collective, while our proposed method should also be applicable to the parameter server architecture.}. There are many optimized algorithms for the all-reduce operation with different number of processes and message sizes \cite{rabenseifner2004optimization}\cite{thakur2005optimization}\cite{hoefler2010toward}. To simplify the problem, we assume that the number of workers is power-of-two, and the peer-to-peer communication cost is modeled as $\alpha+\beta M$ \cite{sarvotham2001connection}, where $\alpha$ is the latency component (or called start-up time), $\beta$ is the communication time per byte, and $M$ is the message size. Without loss of generality, we do not limit the communication model to one specific algorithm. Given $N$ workers, the time cost of all-reduce can be generalized as
\begin{equation}\label{equ:tcomm}
T_{ar}(M)=a+b\times M,
\end{equation}
where $a$ and $b$ are two constants that are not dependent on $M$.
Some well optimized all-reduce algorithms are summarized in Table \ref{table:allreduce}.

\begin{table}[!ht]
		\centering
		\caption{Cost of different all-reduce algorithms}
		\label{table:allreduce}
		\addtolength{\tabcolsep}{-2.2pt}
		\begin{tabular}{|l|c|c|}
			\hline
			All-reduce Algorithm &  $a$ & $b$ \\\hline
			\hline
			Binary tree~\cite{rabenseifner2004optimization} & $2\alpha \log N$ & $(2\beta+\gamma)\log N$ \\\hline
			Recursive doubling~\cite{thakur2005optimization}& $\alpha \log N$ & $(\beta+\gamma)\log N$  \\\hline
			Recursive halving/doubling~\cite{thakur2005optimization}& $2\alpha \log N$ & $2\beta-\frac{1}{N}(2\beta+\gamma)+\gamma$ \\\hline
			Double binary trees~\cite{sanders2009two} & $2\log N$ & $\beta$+$\gamma$ \\\hline
			Ring~\cite{thakur2005optimization} & $2(N-1)\alpha$ & $\frac{2(N-1)}{N}\beta+\frac{(N-1)}{N}\gamma$  \\\hline
		\end{tabular}
\end{table}
With a given hardware configuration (i.e., $N, \alpha, \beta$, and $\gamma$ are fixed), the time cost of the all-reduce operation is a linear function of the message size $M$ with a y-intercept $a$ and a slope $b$. We empirically validate this linear model in Section 6.2.

One important property of WFBP-SGD is that the messages are communicated layer by layer, which means that it needs to invoke many all-reduce operations. In each all-reduce operation, however, there is an extra cost of $a$ which is not related to $M$. Importantly, the linear function with a positive y-intercept value has a property of
\begin{equation}\label{equ:pro}
T_{ar}(M_{1})+T_{ar}(M_{2}) > T_{ar}(M_1+M_2).
\end{equation}
In other words, communicating a single message of size $M_1+M_2$ is more efficient than communicating a message of size $M_1$ and a message of size $M_2$ separately.

\section{Problem Formulation}\label{s:profor}

Eq. (\ref{equ:pro}) indicates that merging the gradients can improve the communication efficiency. If one merges all layers into one layer so that the communication is only invoked once (i.e., the single-layer communication \cite{you2017scaling}), then the overall communication time is minimal. However, the single-layer communication requires all gradients to be calculated first, which prohibits the overlap between communications and computations. Therefore, we would like to merge the layers appropriately so that it not only reduces the communication by merging, but also exploits the pipelining between communications and computations. 

Before formulating the problem, we formally define the concept of merged-gradient layer as follows.
\begin{definition}{(Merged-gradient layer).}
A layer $l$ is called a merged-gradient layer if at the timestamp of $\tau_c^{(l)}$, instead of communicating the gradients of that layer, we merge its gradients to layer $l-1$ and postpone the communication. The operator $\oplus$ defines the gradients merging between two consecutive layers, say $l\oplus (l-1)$. Merging more than two layers is possible by setting consecutive layers into merged-gradient layer.
\end{definition}

\begin{definition}{(Normal layer).}
If a layer $l$ is not a merged-gradient layer, then it is called a normal layer and its gradients will not be merged into layer $l-1$. Its gradients (including those merged from other layers if any) should be communicated as earlier as possible, i.e., when its own gradients have been calculated and the previously scheduled communication has finished.
\end{definition}

There are several properties if layer $l$ is a merged-gradient layer.
\begin{itemize}
    \item $l>1$, since the first layer of the DNN cannot be a merged-gradient layer according to the definition.
    \item There is no communication dedicated for layer $l$, i.e.,
        \begin{equation}\label{ass:1}
            t_{c}^{(l)}=0.
        \end{equation}
    \item The number of updated parameters of layer $l-1$ becomes the summation of that of layer $l$ and layer $l-1$.
        \begin{equation}\label{ass:3}
            p^{(l-1)}=p^{(l-1)}+p^{(l)}.
        \end{equation}
    \item The timestamp when layer $l-1$ can begin the gradient communication is updated to
        \begin{equation}\label{ass:2}
            \tau_c^{(l-1)}=\text{max}\{\tau_c^{(l)}, \tau_b^{(l-1)}+t_b^{(l-1)}\}.
    \end{equation}
\end{itemize}

Intuitively, if merging the gradients of two consecutive layers can save time, then we should merge the two layers. In the following, we discuss a complete set of four cases of computation and communication patterns that may happen during the training process with WFBP for layer $l$. The four cases with potential merging are illustrated in Fig. \ref{fig:pipeline}.
\begin{figure}[!ht]
	\centering
	\begin{subfigure}{0.48\textwidth}
		\includegraphics[width=\linewidth]{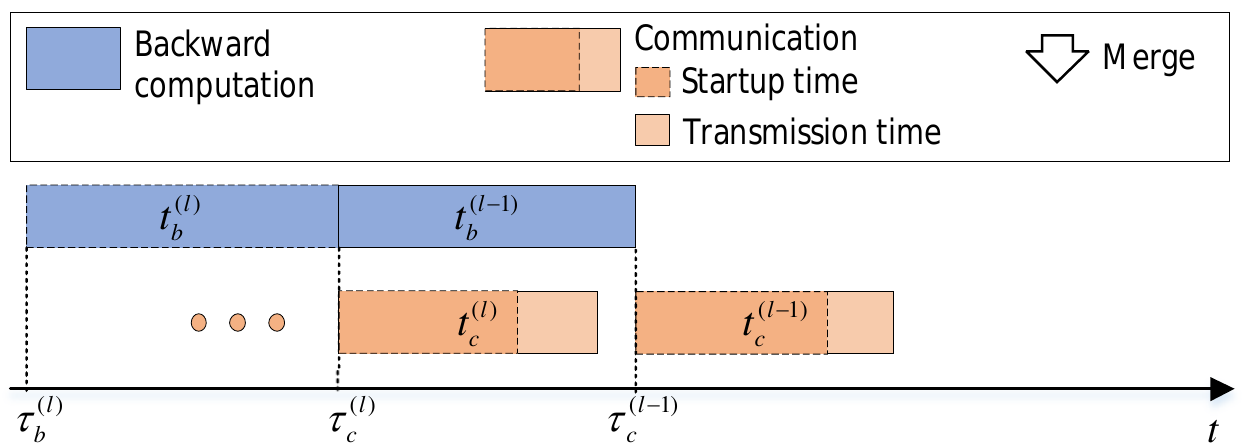}
		\caption{Case 1.}
	\end{subfigure}
	\begin{subfigure}{0.48\textwidth}
	    \vspace{10pt}
		\includegraphics[width=\linewidth]{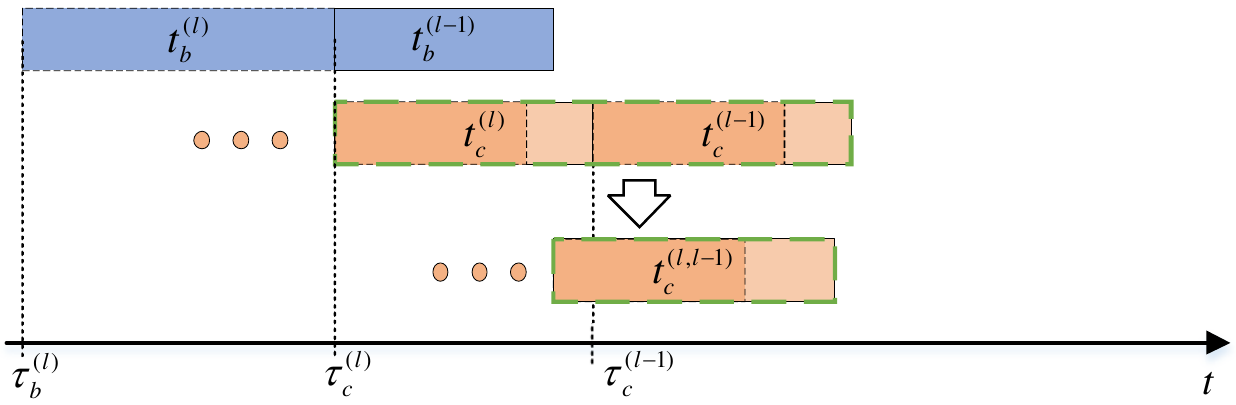}
		\caption{Case 2.}
	\end{subfigure}
	\begin{subfigure}{0.48\textwidth}
	    \vspace{10pt}
		\includegraphics[width=\linewidth]{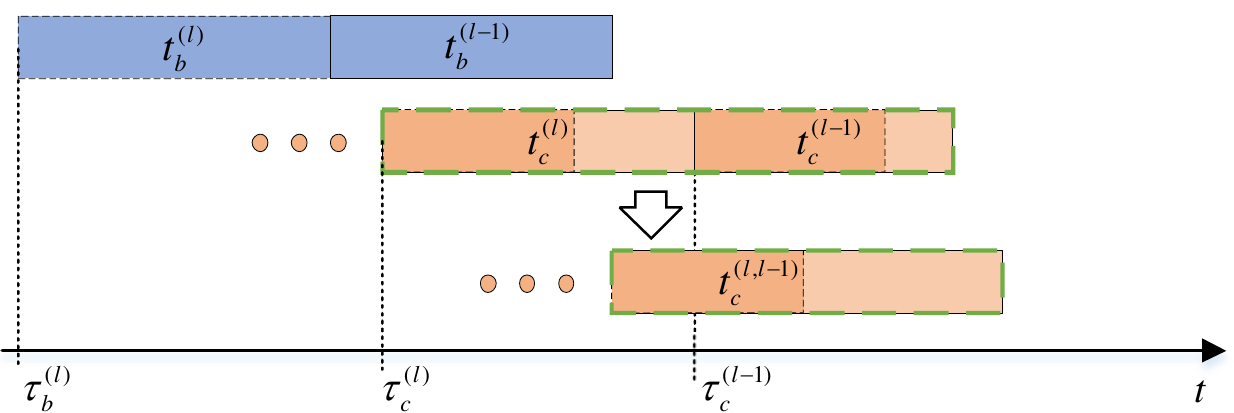}
		\caption{Case 3.}
	\end{subfigure}
	\begin{subfigure}{0.48\textwidth}
	    \vspace{10pt}
		\includegraphics[width=\linewidth]{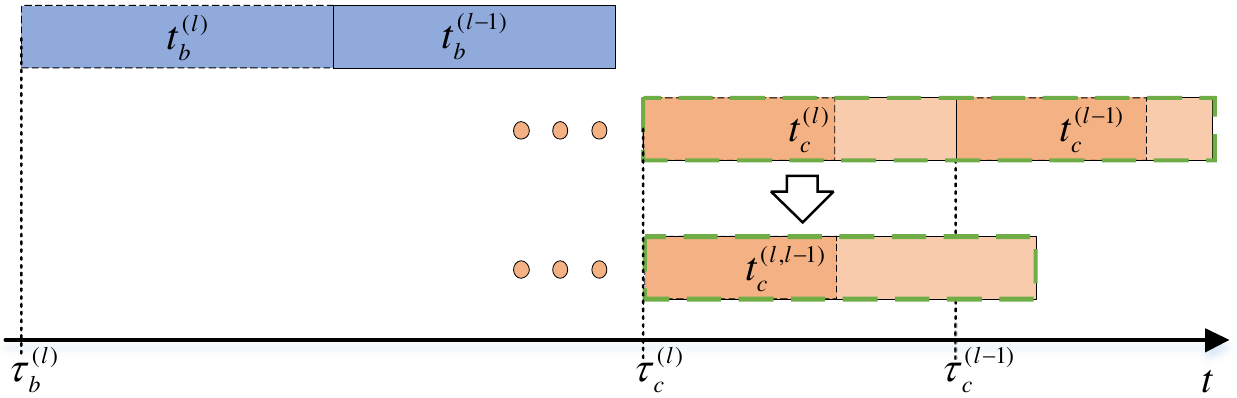}
		\caption{Case 4.}
	\end{subfigure}
	\caption{Four cases of gradient communication at one iteration on layer $l$ in WFBP-SGD. Note that the forward computation is not plotted as it is not related to the pipelining timeline.}
	\label{fig:pipeline}
\end{figure}

\textbf{Case 1}. In the ideal case, the communication of layer $l$ is fully hidden by its previous layer's computation, that is 
\begin{equation}
    \tau_{c}^{(l)}+t_c^{(l)}\leq \tau_{b}^{(l-1)}+t_{b}^{(l-1)}.
\end{equation}
The overhead of gradient communication is totally hidden by computation so that it is not necessary to merge the gradients. 

\textbf{Case 2}. The communication of layer $l$ is partially overlapped with the computation of layer $l-1$, and the communication of layer $l$ begins before the end of the computation of layer $l-1$, that is 
\begin{equation}
\tau_{c}^{(l)}+t_c^{(l)} > \tau_{b}^{(l-1)}+t_{b}^{(l-1)} > \tau_{c}^{(l)}. 
\end{equation}
Without merging, the communication of layer $l$ can immediately begin after the gradients of layer $l$ have been calculated, i.e., $\tau_{c}^{l-1}=\tau_{c}^{(l)}+t_c^{(l)}$. On the other hand, if we want to merge layer $l$ with layer $l-1$, the communication can only happen after the gradients of layer $l-1$ have been calculated. So we should consider whether merging layer $l$ and $l-1$ could bring any benefits or not. As shown in Fig. \ref{fig:pipeline}(b), the merged communication cost takes shorter time to finish, which indicates that the reduced time by merging is greater than the additional waiting time for the gradient computation of layer $l-1$. Formally,
\begin{equation}
\begin{split}
    &\tau_b^{(l-1)}+t_b^{(l-1)}-\tau_c^{(l)} \\
   <& T_{ar}(p^{(l)}+p^{(l-1)})- (T_{ar}(p^{(l)})+T_{ar}(p^{(l-1)})) = a.
\end{split}
\end{equation}
In this case, we prefer to merge the gradients of layer $l$ to layer $l-1$, i.e., making layer $l$ be a merged-gradient layer.

\textbf{Case 3}. In this case, the communication of layer $l$ is also partially overlapped with the computation of $l-1$ as Case 2. However, different from Case 2, the merging operation results in a longer time because the reduced communication time is not as significant as the additional waiting time. To be specific, 
\begin{equation}
    \tau_{c}^{(l)}+t_c^{(l)} > \tau_{b}^{(l-1)}+t_{b}^{(l-1)} > \tau_{c}^{(l)},
\end{equation}
and
\begin{equation}
    \begin{split}
    &\tau_b^{(l-1)}+t_b^{(l-1)}-\tau_c^{(l)} \\
   \geq & T_{ar}(p^{(l)}+p^{(l-1)})- (T_{ar}(p^{(l)})+T_{ar}(p^{(l-1)})) = a.
\end{split}
\end{equation}
Therefore, we would not make layer $l$ be a merged-gradient layer because merging the gradients of layer $l$ to layer $l-1$ will decrease the time efficiency.

\textbf{Case 4}. Very different from the previous cases, there is no overlap between the communication of layer $l$ and the computation of layer $l-1$ as shown in Fig. \ref{fig:pipeline}(d). This happens when the previous communication time is longer than the previous computation time. That is, 
\begin{equation}
    \tau_c^{(l)}\geq \tau_b^{(l-1)}+t_b^{(l-1)}.
\end{equation}
In this case, the communications of layer $l$ and layer $l-1$ do not need to wait for the end of the computation of layer $l-1$; hence merging gradients of layer $l$ to layer $l-1$ dose not introduce any waiting time for the computation, which would obviously reduce the communication time, i.e.,
\begin{equation}
    T_{ar}(p^{(l)}+p^{(l-1)})- (T_{ar}(p^{(l)})+T_{ar}(p^{(l-1)})) = a > 0.
\end{equation}
Thus, we would like to make layer $l$ be a merged-gradient layer in this case.

From the above discussions, we can see that not all gradient merging can bring benefits of reduced iteration time (e.g., Case 3). Therefore, our problem is to find all merged-gradient layers such that the overall iteration time is minimal. Since a layer is either a normal-layer or a merged-gradient layer, we use $l_n$ and $l_m$ to denote the type of  normal-layer and the merged-gradient layer respectively. Let the variable $e^{(l)}$ denote the type of layer $l$ ($l=1,2,...,L$), $e^{(l)}\in \{l_n, l_m\}$. For an $L$-layer DNN model, it can be represented by 
\begin{equation}
    \mathbb{M}=\{[e^{(1)},...,e^{(l)},...,e^{(L)}]|e^{(l)}\in \{l_n, l_m\} \text{ and } 1\leq l\leq L\}.
\end{equation}
Obviously, the number of combinations of normal layers and merge-gradient layers is $|\mathbb{M}|=2^L$. Therefore, our goal is to find an $m\in \mathbb{M}$ such that the iteration time is minimal.

Assuming the linear communication model of Eq. (\ref{equ:tcomm}), the communication time of each layer is represented by
\begin{equation}\label{equ:comm}
t_{c}^{(l)}=T_{ar}(p^{(l)}).
\end{equation}
For a given DNN training with a specific mini-batch size on a hardware environment, the computation time of one iteration can be easily measured at the beginning of training. Since the architecture of the DNN would not change during the training, the feed-forward and backward propagation computation time is very stable \cite{shi2016benchmarking}. That is, $t_b^{(l)}$ is known for $l=1,2,...,L$. However, the beginning timestamp ($\tau_c^{(l)}$) and the communication time ($t_c^{(l)}$) of layer $l$ will be different when $e^{(l)}=l_n$ or $e^{(l)}=l_m$ as we discussed before. Therefore, we generalize the problem as follows.

For a given $L$-layer\footnote{This is also applicable to current DL frameworks like PyTorch, in which the learnable parameters of a layer may be separated as two tensors.} DNN trained with WFBP-SGD on a specific cluster with $P$ workers, we would like to determine $e^{(l)}$ to be $l_n$ or $l_m$ such that the iteration time of training is minimal. Formally, we would like to minimize the iteration time of WFBP-SGD in Eq. (\ref{equ:wfbpiter}), i.e.,
\begin{equation}\label{equ:problem}
    \text{minimize: } t_{iter}=t_c^{(1)}+\max\{\tau_c^{(2)}+t_c^{(2)}, \tau_b^{(1)}+t_b^{(1)}\}.
\end{equation}

\section{Solution: MG-WFBP}\label{s:method}

In this section, we first perform some theoretical analysis on the optimization problem, and then propose an optimal and efficient solution named merged-gradient WFBP (MG-WFPB) to the problem.

\subsection{Theoretical Analysis}
It is obvious that the objective function of Eq. (\ref{equ:problem}) can be rewritten by
\begin{equation}\label{equ:newopt}
\begin{split}
t=&t_c^{(1)}+\text{max}\left\{\tau_c^{(2)}+t_c^{(2)}, \tau_b^{(1)}+t_b^{(1)}\right\}\\
=&T_{ar}(p^{(1)})+\text{max}\left\{\tau_c^{(2)}+T_{ar}(p^{(2)}), \tau_b^{(1)}+t_b^{(1)}\right\}\\
=&T_{ar}(p^{(1)})+\text{max}\left\{ 
\text{max}\{\tau_c^{(3)}+T_{ar}(p^{(3)}), \tau_b^{(2)}+t_b^{(2)}\} \right. \\
 &\left. +T_{ar}(p^{(2)}), \tau_b^{(1)}+t_b^{(1)}\right\}.
\end{split}
\end{equation}

It can be seen that the objective function consists of embedding $\max$ functions from the first layer to the last layer. We first analyze the difference of layer $2$ be a normal layer or a merged-gradient layer, and then we extend it to a general layer $l$ to prove its optimality.

Assume that layers $L,L-1,...,3$ are normal layers, and layer $2$ is a merged-gradient layer, we have $t_c^{(2)}=0$ and $t_c^{(1)}=T_{ar}(p^{(2)}+p^{(1)})$. We plug in these two new values to Eq. (\ref{equ:newopt}) to obtain
\begin{equation}\label{equ:merged}
\begin{split}
\hat{t}&=T_{ar}(p^{(2)}+p^{(1)})+\text{max}\left\{\tau_c^{(2)}, \tau_b^{(1)}+t_b^{(1)}\right\}.
\end{split}
\end{equation}
Compare Eq. (\ref{equ:newopt}) to Eq. (\ref{equ:merged}), we want to find out under what conditions $\hat{t}< t$, i.e., layer $2$ can be a gradient-merged layer. Specifically, we would like to derive the conditions such that 
\begin{equation}
    \begin{split}
        \hat{t}=&T_{ar}(p^{(2)}+p^{(1)})+\text{max}\left\{\tau_c^{(2)}, \tau_b^{(1)}+t_b^{(1)}\right\} \\
        <& t = T_{ar}(p^{(1)})+\text{max}\left\{\tau_c^{(2)}+T_{ar}(p^{(2)}), \tau_b^{(1)}+t_b^{(1)}\right\},
    \end{split}
\end{equation}
which is equivalent to
\begin{equation}\label{equ:ineq}
    \begin{split}
        &b\times p^{(2)}+\text{max}\left\{\tau_c^{(2)}, \tau_b^{(1)}+t_b^{(1)}\right\} \\
        <& \text{max}\left\{\tau_c^{(2)}+T_{ar}(p^{(2)}), \tau_b^{(1)}+t_b^{(1)}\right\}.
    \end{split}
\end{equation}
Since there are two max functions in the above inequality, we need to decompose the max functions. Decomposing the two max functions explicitly corresponds to the four cases we discuss in the previous section. Note that it is impossible that 
$\tau_c^{(2)}+T_{ar}(p^{(2)}) \leq \tau_b^{(1)}+t_b^{(1)}$ and $\tau_c^{(2)} > \tau_b^{(1)}+t_b^{(1)}$ hold simultaneously. Therefore we decompose the two max functions with the following three conditions.

\textbf{Condition 1}. $\tau_c^{(2)}+T_{ar}(p^{(2)}) \leq \tau_b^{(1)}+t_b^{(1)}$. Then $\tau_c^{(2)} \leq \tau_b^{(1)}+t_b^{(1)}$ also holds. The inequality (\ref{equ:ineq}) becomes
\begin{align*}
b\times p^{(2)}+ \tau_b^{(1)}+t_b^{(1)} < \tau_b^{(1)}+t_b^{(1)},
\end{align*}
which obviously does not hold as $b \times p^{(2)}>0$. Therefore, layer $2$ should be a normal layer in this case, since making layer $2$ a merged-gradient layer cannot reduce the iteration time.

\textbf{Condition 2}. The condition is 
\begin{equation}
\tau_{c}^{(2)}+T_{ar}(p^{(2)}) > \tau_{b}^{(1)}+t_{b}^{(1)} > \tau_{c}^{(2)}. 
\end{equation}
We can decompose inequality (\ref{equ:ineq}) to 
\begin{equation}
    \begin{split}
        &b\times p^{(2)}+ \tau_b^{(1)}+t_b^{(1)} \\
        <& \tau_c^{(2)}+T_{ar}(p^{(2)})=\tau_c^{(2)}+a+b\times p^{(2)},
    \end{split}
\end{equation}
which is equivalent to 
\begin{equation}\label{equ:inq-c2}
    \begin{split}
        \tau_b^{(1)}+t_b^{(1)}<\tau_c^{(2)}+a.
    \end{split}
\end{equation}
So if inequality (\ref{equ:inq-c2}) is true, then we can make layer $2$ a merged-gradient layer to save the iteration time; otherwise we make it a normal layer. 

\textbf{Condition 3}. The condition is 
\begin{equation}
\tau_{c}^{(2)}+T_{ar}(p^{(2)}) > \tau_{c}^{(2)}> \tau_{b}^{(1)}+t_{b}^{(1)}.
\end{equation}
We decompose inequality (\ref{equ:ineq}) to 
\begin{equation}
    \begin{split}
        b\times p^{(2)}+\tau_c^{(2)}<\tau_c^{(2)}+T_{ar}(p^{(2)}).
    \end{split}
\end{equation}
It is equivalent to
\begin{equation}
    \begin{split}
        b\times p^{(2)}+\tau_c^{(2)}<\tau_c^{(2)}+a+b\times p^{(2)},
    \end{split}
\end{equation}
which is obviously true as $a>0$. Therefore, under this condition, we prefer to make layer $2$ a merged-gradient layer.

To summarize, under Condition 2 with inequality (\ref{equ:inq-c2}) and Condition 3, making layer $2$ a merged-gradient layer can reduce the iteration time. Now we extend the above analysis to a general layer $l$ and $l>1$. When we just consider the end time of layer $l-1$, making layer $l$ be a merged-gradient layer if Condition 2 with inequality (\ref{equ:inq-c2}) holds or Condition 3 holds will reduce the end time of layer $l-1$. Thus, we have the following lemma.
\begin{lemma}\label{lemma:mergedlayer}
Given an $L$-layer DNN which is trained with WFBP-SGD in a cluster of $N$ workers, if the gradient communication is done through all-reduce, layer $l>1$ should be a merged-gradient layer to reduce the iteration time if and only if
\begin{equation}\label{equ:lemma-inq1}
    \tau_b^{(l-1)}+t_b^{(l-1)} < \tau_c^{(l)}+a.
\end{equation}
\begin{proof}
As we discussed in the above three conditions, if Condition 2 together with inequality (\ref{equ:inq-c2}) or Condition 3 holds, layer $l$ should be a merged-gradient layer to reduce the iteration time, otherwise it should be a normal layer. The combination of Condition 2 together with inequality (\ref{equ:inq-c2}) and Condition 3 is 
\begin{equation}
    \tau_b^{(l-1)}+t_b^{(l-1)} < \tau_c^{(l)}+a,
\end{equation}
which concludes the proof.
\end{proof}
\end{lemma}

From Lemma \ref{lemma:mergedlayer}, it is seen that whether layer $l$ should be a merged-gradient layer or not depends on the end of computation time of layer $l-1$ (i.e., $\tau_b^{(l-1)}+t_b^{(l-1)}$) and its own beginning time of communication (i.e., $\tau_c^{(l)}$). Thus, the communications of higher layers are not affected by the lower layers, while the lower layers are affected by the higher ones as the lower layer can only begin after the higher layers have finished. If layer $l$ is a normal layer, we can continue to determine layer $l-1$ by checking the above three conditions. If layer $l$ is a merged-gradient layer, layer $l-1$ has earlier end time according to the benefit of the merged-gradient layer. Again we also continue to determine the type of layer $l-1$ as the same way of layer $l$, which results in a recursive way from layer $L$ to layer $2$. Consequently, we determine the last layer $L$ whether it can be a merged-gradient layer or a normal layer, and then determine layer $L-1$, and finally to layer $2$ to find the final solution $m\in \mathbb{M}$ such that Eq. (\ref{equ:newopt}) is minimal. 

\begin{theorem}\label{theorem:opt}
Given an $L$-layer DNN which is trained with WFBP-SGD in a cluster of $N$ workers, if the gradient communication is done through all-reduce, one can find $m\in \mathbb{M}$ such that the iteration time is minimal, and
\begin{equation}\label{the:solution}
m=[e^{(L)}, e^{(L-1)},...,e^{(1)}],
\end{equation}
where 
\begin{equation}\label{equ:layertype}
    e^{(l)}=
        \begin{cases}
        l_m & \text{if }\tau_b^{(l-1)}+t_b^{(l-1)} < \tau_c^{(l)}+a \text{ and } l>1\\
        l_n & \text{otherwise}
        \end{cases}
\end{equation}
for $1 \leq l \leq L$.
\end{theorem}
\begin{proof}
A layer $l$ is either a merged-gradient layer or a normal layer. According to Lemma \ref{lemma:mergedlayer}, for $l>1$ and $\tau_b^{(l-1)}+t_b^{(l-1)} < \tau_c^{(l)}+a$, $e^{(l)}=l_m$ has shorter time than $e^{(l)}=l_n$. For $l=1$ or $\tau_b^{(l-1)}+t_b^{(l-1)} \geq \tau_c^{(l)}+a$, $e^{(l)}=l_n$ has shorter time than $e^{(l)}=l_m$. Consequently, if $m=[e^{(L)}, e^{(L-1)},...,e^{(1)}]$ and $e^{(l)}$ is assigned by Eq. (\ref{equ:layertype}), then changing the merged-gradient layers to normal layers or changing the normal layers to merged-gradient would bring longer iteration time, which conclude the proof.
\end{proof}

\subsection{Algorithms}
Assume that the $N$-node cluster is connected by an interconnection with a bandwidth $B$, we can measure the all-reduce cost with respect to message size to derive the parameter $a$ and $b$ in Eq. (\ref{equ:tcomm}). Therefore, we can estimate the communication time of all-reduce for any message size. For the backward computation time, we can also benchmark for a particular GPU at the beginning of training. Thus, $t_f$, $t_b^{(l)}$ and $t_c^{(l)}$, where $1 \leq l \leq L$, are known. According to Theorem \ref{theorem:opt}, we drive the algorithm to find $m$ as shown in Algorithm \ref{algo:mgbp}.

\begin{algorithm}[h]
	\caption{Find optimal $m\in \mathbb{M}$}\label{algo:mgbp}
 	\small
		\textbf{Input: }$a$, $b$, $L$, $\bm{t_b}[1...L]$, $\bm{p}=[p^{(1)},p^{(2)},...,p^{(L)}]$.\\
		\textbf{Output: $m$}
	\begin{algorithmic}[1]
		\State Initialize $\bm{t_c}[1...L]$; // Communication time cost
		\State Initialize $\bm{\tau_b}[1...L]$; // Backward computation start time
		\State Initialize $m[1...L]=\{l_n\}$; // Initialize all layers be normal layers
		\For{$l=1\rightarrow L$}
			\State $\bm{t_c}[l]=a+b\times \bm{p}[l]$;
		\EndFor
		\State $\bm{\tau_b}[L]=0$;
		\For{$l=L-1\rightarrow 1$}
			\State $\bm{\tau_b}[l]$ = $\bm{\tau_b}[l+1]$ + $\bm{t_b}[l+1]$;
		\EndFor
		\State $\bm{\tau_c}$=\Call{CalculateCommStart}{$\bm{t_c}, \bm{t_b}, \bm{\tau_b}, L$}; 
		\For{$l=L\rightarrow 2$}
		\If{$\bm{\tau_b}[l-1]+\bm{t_b}[l-1]-\bm{\tau_c}[l] < a$} // Eq. (\ref{equ:layertype})
		\State \Call{Merge}{$\bm{\tau_b}, \bm{t_c}, \bm{p}, l$};
		\State $\bm{\tau_c}$=\Call{CalculateCommStart}{$\bm{t_c}, \bm{t_b}, \bm{\tau_b}, L$};
		\State $m[l]=l_m$; // Make $l$ be the merged-gradient layer
		\EndIf
		\EndFor
		\State Return $m$;

		\Procedure{Merge}{$\bm{\tau_b}, \bm{t_c}, \bm{p}, l$}
    		\State $\bm{t_c}[l]=0$;
    		\State $\bm{p}[l-1]=\bm{p}[l-1]+\bm{p}[l]$;
    		\State $\bm{t_c}[l-1]=a+b\times \bm{p}[l-1]$;
		\EndProcedure

		\Procedure{CalculateCommStart}{$\bm{t_c}, \bm{t_b}, \bm{\tau_b}, L$}
		\State Initialize $\bm{\tau_c}[1...L]$; // Communication start time
		\State $\bm{\tau_c}[L]=\bm{\tau_b}[L]+\bm{t_b}[L]$;
		\For{$l=L-1\rightarrow 1$}
		\State $\bm{\tau_c}[l]=\text{max}\{\bm{\tau_c}[l+1]+\bm{t_c}[l+1], \bm{\tau_b}[l]+\bm{t_b}[l]\}$;
		\EndFor
		\State \text{Return } $\bm{\tau_c}$;
		\EndProcedure
	\end{algorithmic}
\end{algorithm}

The algorithm first (line 1-8) initializes the layer-wise gradient communication cost $t_c^{(l)}$, the computation start time $\tau_b^{(l)}$  according to Eq. (\ref{equ:tcomm}) and Eq. (\ref{equ:startcomp}) respectively with system settings and benchmarks in the first several iterations. Then (line 9, line 20-25) the layer-wise start time of communication is calculated based on Eq. (\ref{equ:startt}). After that (line 10-14), the merged-gradient layers are found according to Eq. (\ref{equ:layertype}), in which if there is a layer found as a merged-gradient layer, the communication time of its previous layer should be updated (line 16-19) according to Eq. (\ref{ass:1}), Eq. (\ref{ass:3}) and Eq. (\ref{ass:2}).

The proposed algorithm has a time complexity of $O(L^2)$. For a merged-gradient layer, the algorithm needs to re-calculate the start time of communication of each layer, which is an $O(L)$ search, and it has maximal $L-1$ merged-gradient layers, so the time complexity of the algorithm is $O(L^2)$. Since the algorithm is a one-time calculation at the beginning of the training and it needs not to be re-calculated during the training process, the overhead of finding $m\in \mathbb{M}$ has no side-effect to the training performance.

\begin{algorithm}[h]
	\caption{MG-WFBP S-SGD at worker $g$}\label{algo:gewfbp}
	
	\textbf{Input: } $\bm{D}=[\{X_1, y_1\},...,\{X_n, y_n\}]$, $I$, $net$, $N$, $bs$\\
	\textbf{Output: $\bm{W}=[W^{(1)}, W^{(2)},...W^{(L)}]$}
	\begin{algorithmic}[1]
		\small
		\State Initialize a shared and synchronized queue $Q$;
		\State Obtain the parameter size $\bm{p}[1...L]$ from $net$;
		\State Allocate memories $\bm{W}$;
		\State Initialize $\bm{W}$ in all accelerators;
		\If{rank == 0}
		    \State Benchmark several iterations to achieve $\bm{t_b}[1...L]$;
		    \State Get $m$ from Algorithm \ref{algo:mgbp};
		\EndIf
		\State Bcast($m$, root=0); // Broadcast the optimal solution to all workers
		\State \Call{AsyncHandleCommunication}{$Q, m$};
		\For{$i=1\rightarrow I$}
		\State Sample a mini-batch of data from $D$ to $d$;
		\State \Call{AsyncHandleComputation}{$Q,d,L$};
		\State WaitForLastCommunicationFinished();
		\State $\bm{W}=\bm{W}-\eta\cdot\nabla \bm{W}$,
		\EndFor
		\State NotifyFinished(); // Set $isRunning$ to false
		
		\Procedure{AsyncHandleComputation}{$Q,d,L$}
		\State $o=d$;
		\For{$l=1\rightarrow L$}
		\State $o$=FeedForward($l,o$);
		\EndFor
		\For{$l=L\rightarrow 1$}
		\State BackwardPropagation($l$);
		\State $Q.\text{push}(l)$;
		\EndFor
		\EndProcedure
		
		\Procedure{AsyncHandleCommunication}{$Q, m$}
		\State Initialize $lb$; // layerBuffer
		\While{\textit{isRunning}}
		\State $l=Q.\text{pop()}$;
		\State $lb$.push($l$);
		\If{$m[l] == l_n$}
		    \State SynchonizedAllReduce($lb$);
		    \State $lb.$clear();
		\EndIf
		\If{$l=1$}
		    \State NotifyLastCommunicationFinished();
		\EndIf
		\EndWhile
		\EndProcedure
	\end{algorithmic}
	
\end{algorithm}
We denote the WFBP algorithm integrated with the optimal solution $m$ derived from Algorithm \ref{algo:mgbp} as MG-WFBP. In MG-WFBP, the merged-gradient layers should be communicated with their previous layers. As a result, MG-WFBP achieves the minimal iteration time of S-SGD under known DNNs and system configurations. The algorithm of MG-WFBP S-SGD is shown in Algorithm \ref{algo:gewfbp}. For each worker, the algorithm first (line 1-7) initializes related variables and calculates $m\in \mathbb{M}$ by using Algorithm \ref{algo:mgbp}. Then the root worker (rank 0) broadcasts (line 8) the solution $m$ to all other workers. Line 9 starts a communication thread, and the thread reads the layer number from the shared queue $Q$ and decides whether its gradients should be communicated (line 24-32). After that (line 10-14), it starts the loop of iteration, and iteratively (line 16-22) reads data to do feed forward operations and backward propagation followed by pushing the layer number into the shared queue. Finally, the algorithm notifies a message of \textit{isRunning=false} to finish the training.

\subsection{Applicability to Parameter Server}
The MG-WFBP algorithm relies on the measurement of layer-wise computing time and the modeling of gradient aggregation to obtain the optimal merging strategy. Regarding the parameter server (PS) architecture, the layer-wise computation is the same as that of all-reduce, while the gradient aggregation process becomes a two-phase operation: (1) the workers push gradients to PS; and (2) the workers pull the aggregated gradients (or the latest model parameters) from PS. Each direction of communication (i.e., push or pull) can also be modeled by Eq. (\ref{equ:tcomm}). Therefore, the theoretical analysis of the optimal merge with Theorem~\ref{theorem:opt} holds and our MG-WFBP is applicable to the PS architecture.

\section{System Implementation}\label{s:system}
As shown in Algorithm \ref{algo:gewfbp}, to implement MG-WFBP, our system is required to be equipped with three main features. First, the system needs to measure the backward propagation computation time of each layer (i.e., $t_b^{(l)}$) for any configured deep neural networks. Second, the backward computation and gradient aggregation should be executed in parallel to pipeline communications and computations. Third, the merging operation of the merged-gradient layer should be efficient. It is non-trivial to implement the above three functions in current state-of-the-art DL frameworks (e.g., TensorFlow \cite{abadi2016tensorflow} and PyTorch \cite{pytorch2019}) which exploit a directed acyclic graph (DAG) to represent computing operations during training. Considering that PyTorch becomes more and more popular due to its easy-to-use Pythonic programming style and high performance operators \cite{pytorch2019}, in this section we describe the implementation of MG-WFBP algorithm atop PyTorch.

\subsection{Time Measurement of Backward Propagation}
When deploying the DAG to GPUs in PyTorch, different operators could be executed concurrently due to the execution nature of CUDA streams \cite{nvidia2011cuda}. Therefore, during the backward propagation, the gradients of different variables could be calculated concurrently on the same GPU such that the time measurement of each variable is not straightforward. To correctly collect the backward propagation time, we design a lightweight profiling tool for backward propagation in PyTorch with sequential execution of different variables. For each tensor that has gradients, we synchronize the tensor after it finishes invoking the gradient computation with CUDA synchronization ($torch.cuda.synchronize$). Consequently, we can collect the time interval of gradient computation between two nearby tensors, and the two nearby tensors with gradients should be from a single layer or from two nearby layers. The measurement can be automatically executed in our MG-WFBP training algorithm during the first several iterations.

\subsection{Parallelism between Gradient Computation and Aggregation}
It is known that current DL frameworks provide Python APIs for end-users. In general, one can use multi-threading or multi-processing to make gradient computation and aggregation executed on two different threads or processes. On one hand, however, there exists the GIL problem \cite{beazley2010understanding} in multi-threading of Python, which would result in very poor performance when paralleling two computation tasks. On the other hand, the multi-processing mechanism requires the memory copy between two processes as the gradients are updated every iteration. Since multiple processes cannot share the GPU memory address, when a process needs to copy its GPU data to another process, it needs to copy the data to host memory and then to another process, which causes performance degradation. To avoid the GIL problem in Python and memory copy between processes, we implement the gradient aggregation in a C++ daemon thread, and the original training codes are kept unchanged and the original training process (forward and backward computation) is running in the main thread. The C++ daemon thread well addresses the GIL problem in Python, and it can share the data of gradient with the main thread so that no memory copy is required. The architecture is shown in Fig. \ref{fig:sysarch}.
\begin{figure}[!h]
	\centering
	\includegraphics[width=\linewidth]{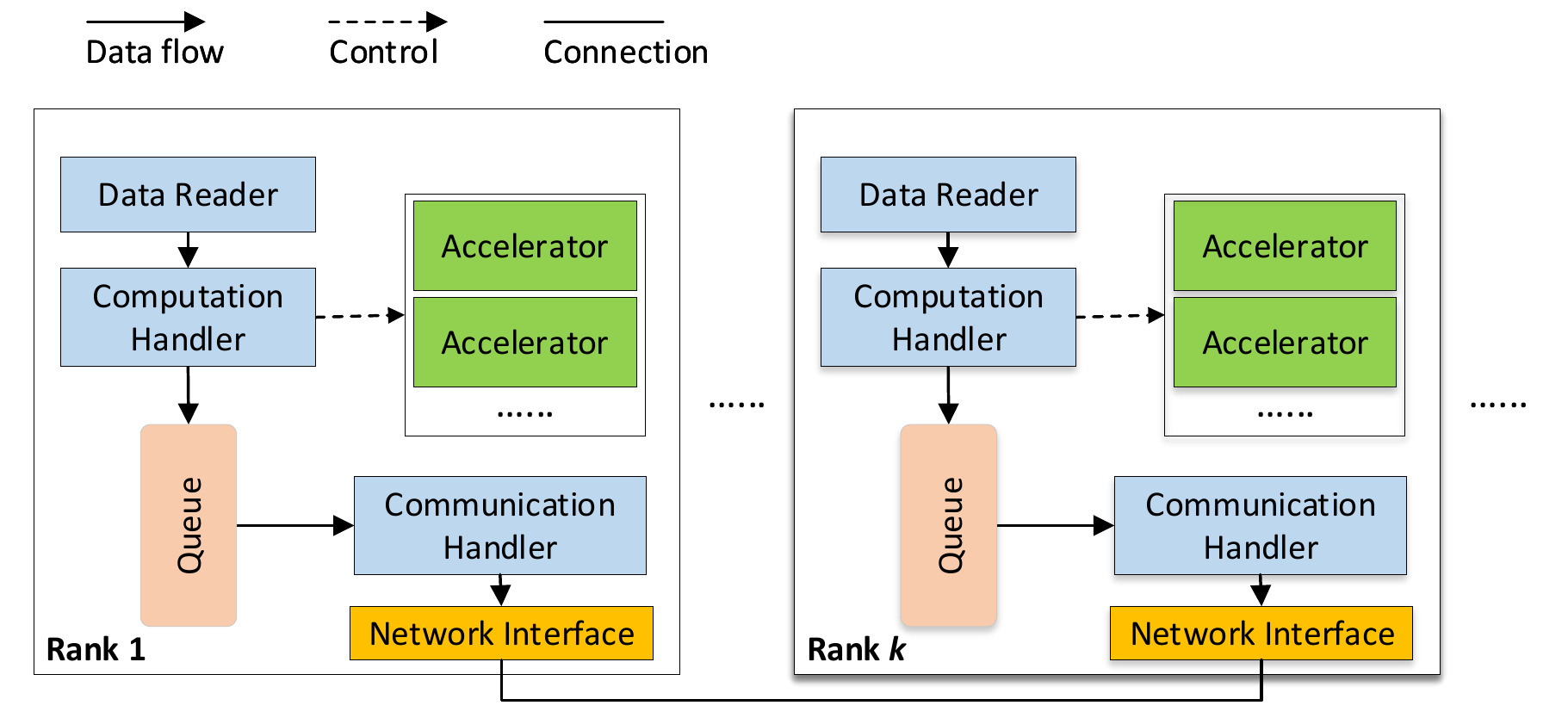}
	\caption{Overview of system architecture.}
	\label{fig:sysarch}
\end{figure}

The key component is the communication handler implemented in a C++ daemon thread. It fetches the gradients from the queue that is shared with the computation handler. During the backpropagation, the computation handler in the main thread puts the gradients from the normal layers to the shared queue, while the communication handler pops the gradients from the queue and communicates with other workers with an all-reduce operation.

\subsection{Efficient Gradient Merging}
For every iteration, we need to copy two layers' gradient data to a single segment of continuous memory if the current layer is a merge-gradient layer. We pre-allocate all memory for merged-gradient layers. For example, we assume layer 2 is a merge-gradient layer, which has $p^{(2)}$ parameters, and layer 1 has $p^{(1)}$ parameters. Note that layer 1 and layer 2 have different tensors so that the memory for these two tensors may not be continuous. Then we allocate a buffer whose size is $(p^{(2)}+p^{(1)})\times BytesPerElement$, where $BytesPerElement$ is 4 for single precision floats and 2 for half precision floats (e.g., Mixed precision training). Therefore, for every merged-gradient layers and their preceded normal layers, there exist pre-allocated buffers. When any buffer is full, the gradient aggregation thread invokes the all-reduce operation. The pre-allocated buffers for merged-gradient layers consume the same memory size as the model parameters, but it is relatively small as compared to the occupied memory of the temporary outputs of hidden layers. In PyTorch, the data copy between GPU tensors is fast as it just needs to copy data in GPU memory without copying back to host memory. For example, Nvidia Tesla V100 GPU delivers a peak memory bandwidth of 900GB/s.

\section{Experimental Studies}\label{s:eval}

\subsection{Experimental Settings}
We conduct extensive experimental studies to show the effectiveness of MG-WFBP. Our test-beds contain three GPU clusters with 10Gbps Ethernet (10GbE) and 56Gbps InfiniBand (56GbIB) interconnections. One is an 8-node Nvidia Tesla K80 cluster which has a total of 16 GK210 GPUs (one Tesla K80 card contains two GK210 GPUs), and the 8 nodes are connected by 10GbE; the other two are 4-node Nvidia Tesla V100 clusters, in which each node contains 4 GPUs, resulting in a total of 16 GPUs, and the 4 nodes are connected with 10GbE and 56GbIB. The cluster settings are listed in Table \ref{table:clusters}. 
\begin{table}[!ht]
	\centering
	\caption{The hardware and software settings on one node.}
	\label{table:clusters}
	\begin{tabular}{|l|c|c|c|}
		\hline
		 & 	Cluster 1 & Cluster 2 & Cluster 3 \\\hline
		\hline
		\# of Nodes & 	8 &  \multicolumn{2}{c|}{4}  \\\hline
		GPU (Nvidia) & 	Tesla K80 &  \multicolumn{2}{c|}{Tesla V100 PCIe x4}  \\\hline
		Network		&	10GbE & 10GbE & 56GbIB \\\hline
		PCIe        &  \multicolumn{3}{c|}{PCI Express Gen3 x16}  \\\hline
		CPU	(Intel) &	Xeon E5-2650v4 Dual & \multicolumn{2}{c|}{Xeon E5-2698v3 Dual}\\\hline
		Memory		&	\multicolumn{3}{c|}{256 GB} 	\\\hline
		OS & CentOS-7.2 & \multicolumn{2}{c|}{Ubuntu 16.04} \\\hline
		Software & CUDA-8.0 & \multicolumn{2}{c|}{CUDA-10.0} \\\cline{2-4}
		         & OpenMPI-3.1.1 & \multicolumn{2}{c|}{OpenMPI-4.0.0} \\\cline{2-4}
		         & NCCL-2.2.12 & \multicolumn{2}{c|}{NCCL-2.3.7} \\\hline
	\end{tabular}
\end{table}

First, we conduct experiments to measure the communication performance on the three clusters. Second, we evaluate the end-to-end training wall-clock time on representative real-world DNN models including GoogleNet \cite{szegedy2015going}, ResNet-50/152 \cite{he2016deep}, DenseNet-161/201 \cite{huang2017densely} and Inception-v4 \cite{szegedy2017inception} with the ImageNet dataset ILSVRC-2012 \cite{deng2009imagenet} which contains about $1.28$ million training images and $50,000$ validation images of $1,000$ categories. The resolution of the input images is $224\times224$. The training settings of DNN models are listed in Table \ref{table:dnns}. 

\begin{table}[!ht]
	\caption{DNNs for evaluation.}
	\label{table:dnns}
	\begin{tabular}{|l|l|l|l|l|}
		\hline
		Model& 	\# Tensors &\# Parameters  & \# MACs & Batch Size   \\\hline
		\hline
		GoogleNet &59 & \textasciitilde 13M	& 1.43G & 64\\\hline
		ResNet-50 &161 & \textasciitilde 25.5M	& 3.9G & 32 \\\hline
		ResNet-152 & 467 & \textasciitilde 60.1M & 11.61G & 128 \\\hline
		DenseNet-161 &484 & \textasciitilde 28.6M	& 7.85G & 64 \\\hline
		DenseNet-201 &604 & \textasciitilde 20M	& 4.39G & 64 \\\hline
		Inception-v4 &449 & \textasciitilde 42.6M & 6.16G & 128 \\\hline
	\end{tabular}
	Note: \# MACs indicates the number of multiply and accumulates in the forward calculation with a batch size of 1.
\end{table}

\subsection{Measurement of All-reduce Communication}
\begin{figure*}[!ht]
	\centering
    \begin{subfigure}{0.32\textwidth}
		\includegraphics[width=\linewidth]{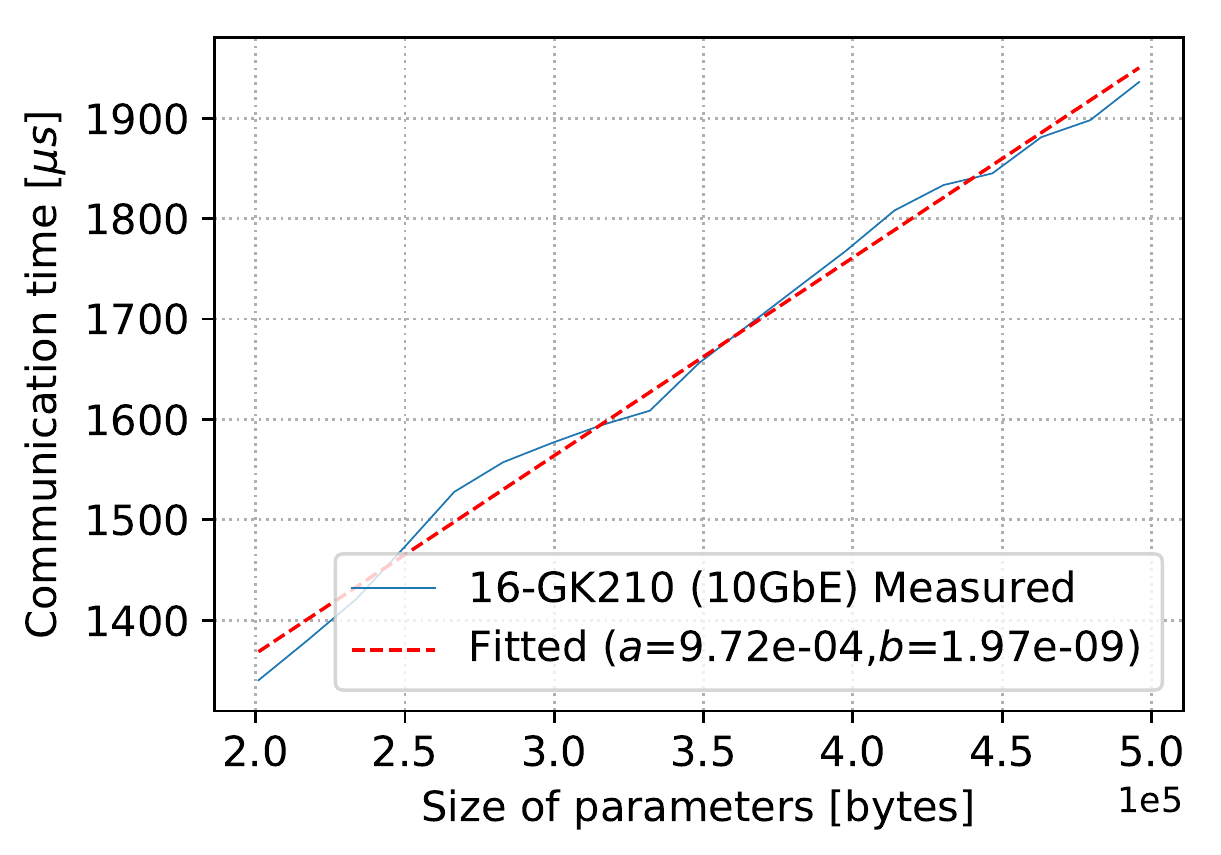}
		\caption{Cluster 1}
	\end{subfigure}
    \begin{subfigure}{0.32\textwidth}
		\includegraphics[width=\linewidth]{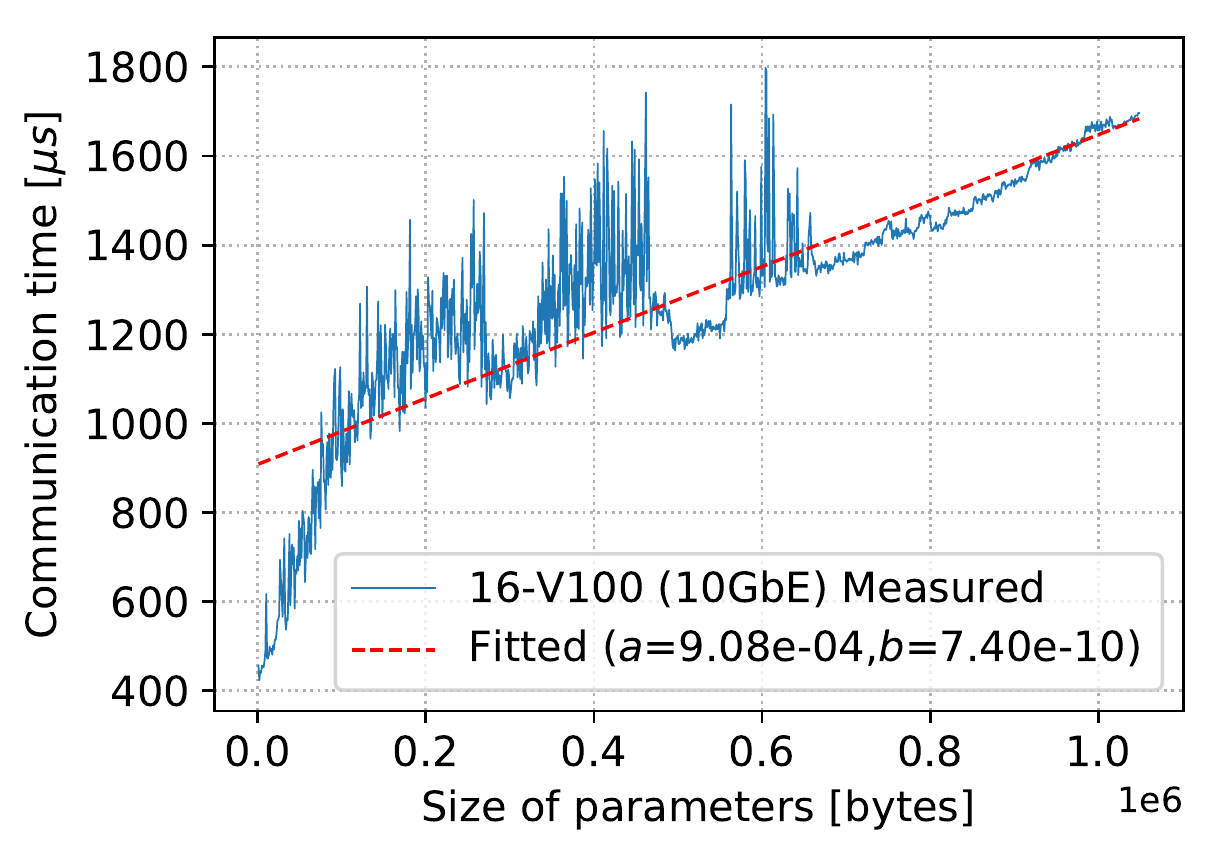}
		\caption{Cluster 2}
	\end{subfigure}
	\begin{subfigure}{0.32\textwidth}
		\includegraphics[width=\linewidth]{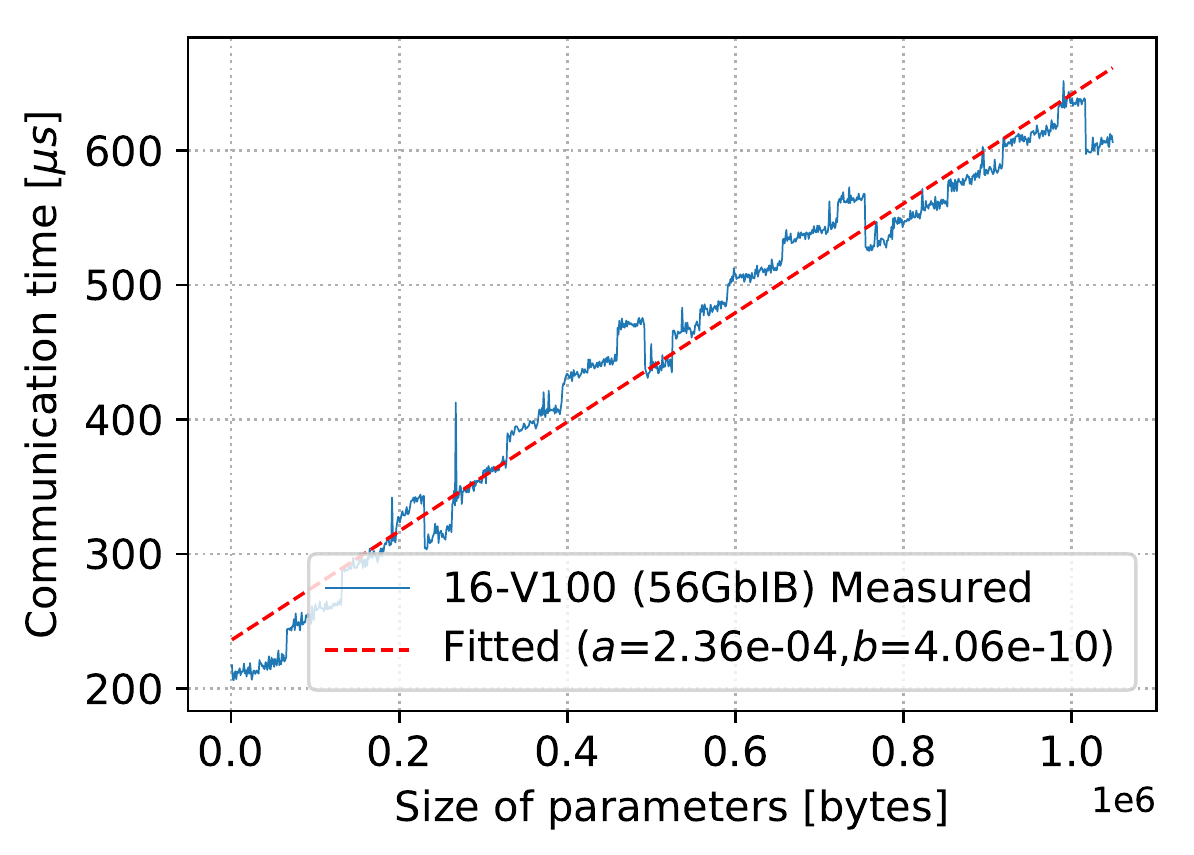}
		\caption{Cluster 3}
	\end{subfigure}
	\caption{The communication time of all-reduce along with the size of parameters on three different clusters. (a) Cluster 1 with $a=9.72\times 10^{-4},b=1.97\times 10^{-9}$; (b) Cluster 2 with $a=9.08\times 10^{-4},b=7.4\times 10^{-10}$; (c) Cluster 3 with $a=2.36\times 10^{-4},b=4.06\times 10^{-10}$.}
	\label{fig:commoverhead}
\end{figure*}

To verify the communication model in Eq. (\ref{equ:tcomm}) empirically, we first present some foregone results of the time of the all-reduce operation in the three configured clusters. The measured time of all-reduce under cluster 1, cluster 2 and cluster 3 are shown in Fig. \ref{fig:commoverhead}(a), \ref{fig:commoverhead}(b) and \ref{fig:commoverhead}(c) respectively. Take the size of parameters ($4p$ in single precision floating points) as the variable, we can see that the startup overheads (e.g., $2(N-1)\times \alpha$ in the ring-based all-reduce algorithm) are $972\mu s$, $908\mu s$ and $236\mu s$ on cluster 1, cluster 2 and cluster 3 respectively. 

We also show the statistical distributions of the layer-wise tensor size in different DNNs in Fig. \ref{fig:tensordistribution}, which shows that a large proportion of tensors are with small number of gradients. For example, ResNet-152 has 150 tensors whose size is 1024 bytes (in 32-bit precision), and DenseNet-161 has 160 tensors whose size is 768 bytes (in 32-bit precision).
\begin{figure}[!h]
	\centering
	\includegraphics[width=0.95\linewidth]{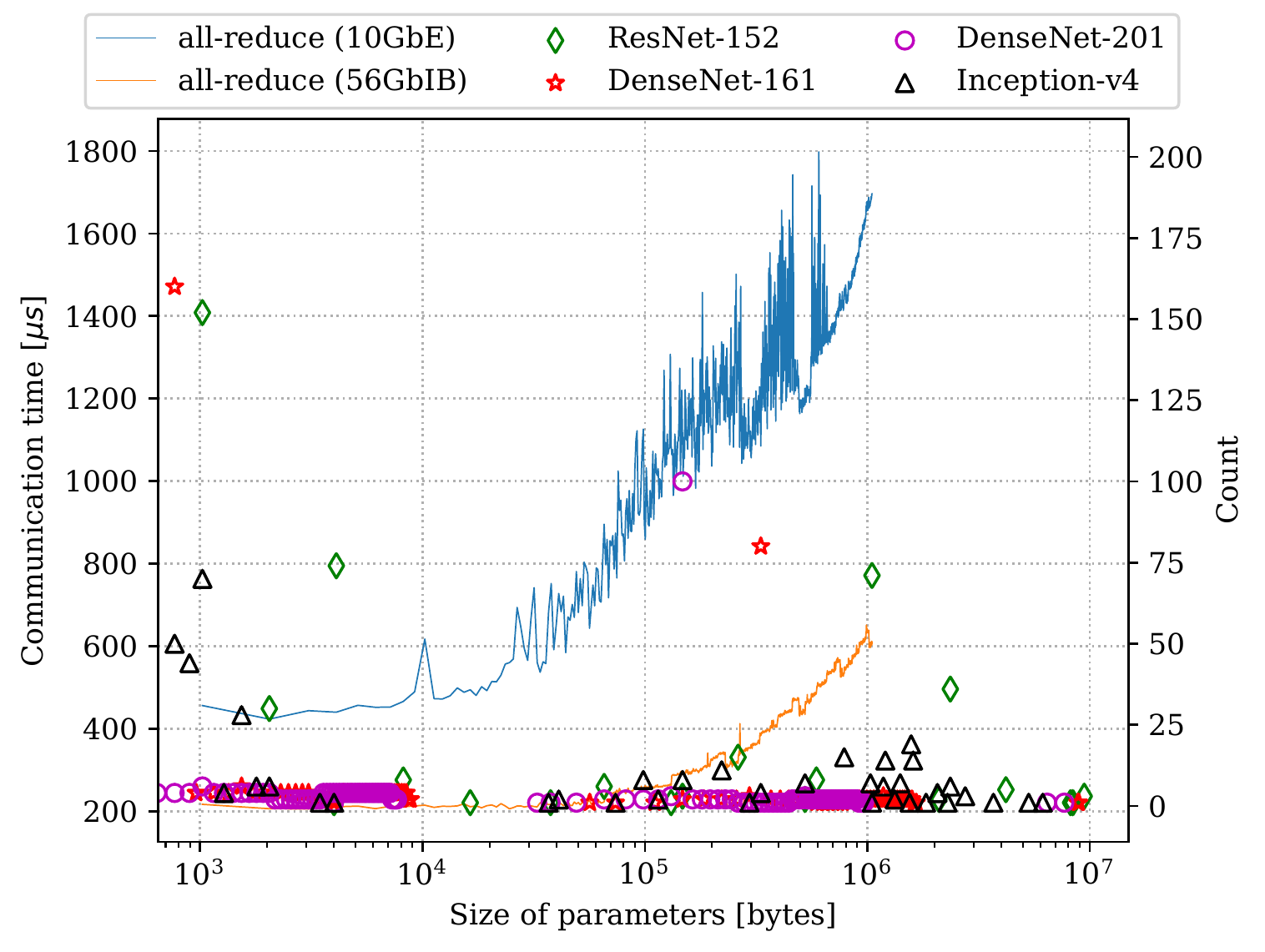}
	\caption{Tensor size distribution. The two curves are the all-reduce communication time with respect to tensor size, and the scatter markers indicate the number of tensors that have a specific size in a DNN.}
	\label{fig:tensordistribution}
\end{figure}

\subsection{Real-world Experiments}
We implement WFBP \cite{awan2017s}\cite{zhang2017poseidon}, single-layer communication Sync EASGD (SyncEASGD) \cite{you2017scaling} and our proposed MG-WFBP with PyTorch and OpenMPI, and test the performance across two 16-GPU (K80 and V100) clusters with 10GbE and 56GbIB. We also compare the scaling efficiencies with TensorFlow. The compared TensorFlow version is at v1.3, and it uses parameter servers to do S-SGD using the official benchmark script\footnote{https://github.com/tensorflow/benchmarks}. We also run 13 epochs to verify the convergence of the model training, in which $50,000$ images are used to test the top-1 accuracy.

\subsubsection{Results on Cluster 1}
\begin{figure}[!h]
	\centering
    \begin{subfigure}{0.24\textwidth}
		\includegraphics[width=\linewidth]{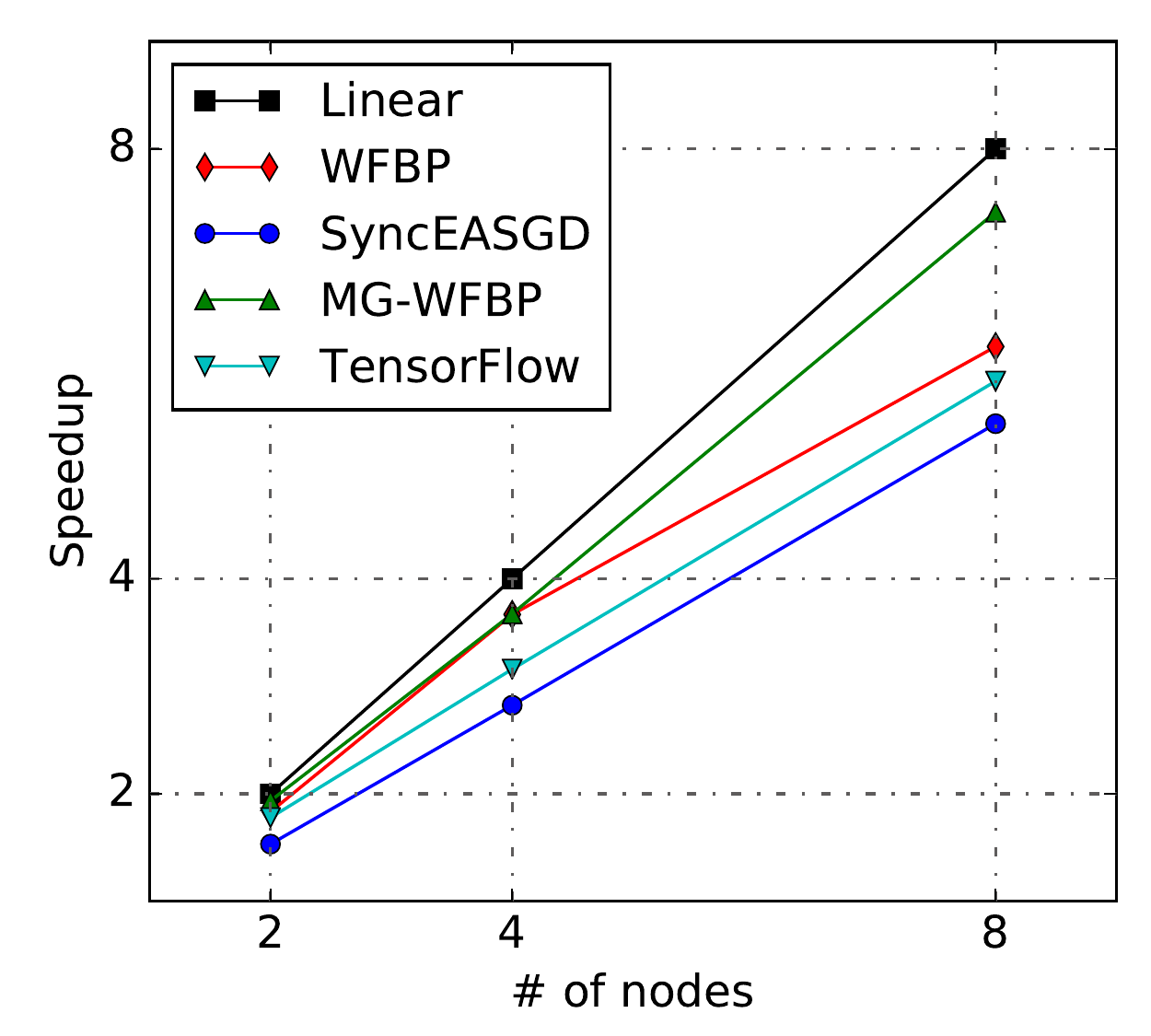}
		\caption{Speedup}
	\end{subfigure}
	\begin{subfigure}{0.24\textwidth}
		\includegraphics[width=\linewidth]{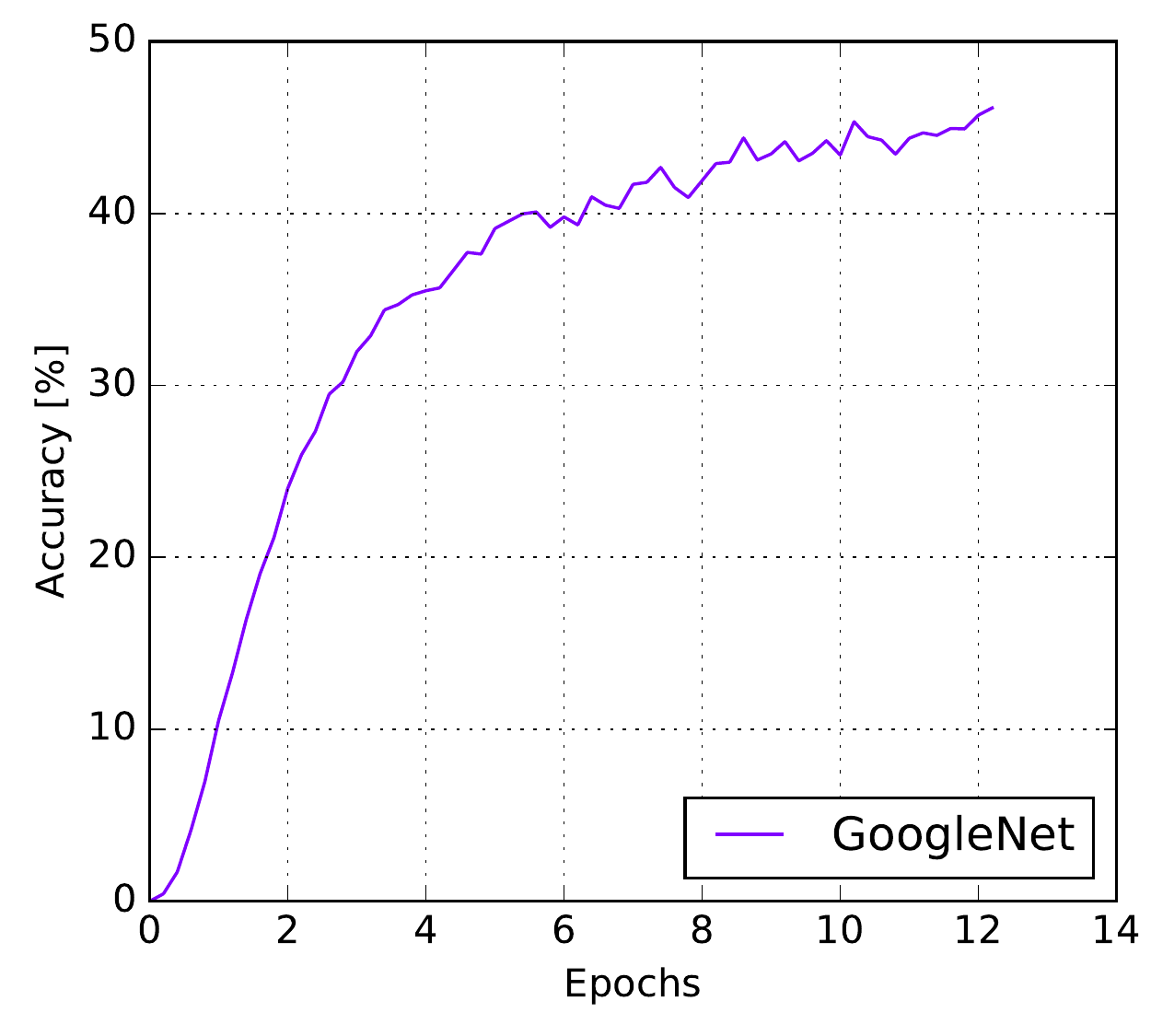}
		\caption{Top-1 validation accuracy}
	\end{subfigure}
	\caption{The performance of GoogleNet on the K80 cluster with 10GbE. The baseline of the speedup of SGD is on a single machine with 2 GPUs.}
	\label{fig:realresultsgooglenet}
\end{figure}

\begin{figure}[!h]
	\centering
    \begin{subfigure}{0.24\textwidth}
		\includegraphics[width=\linewidth]{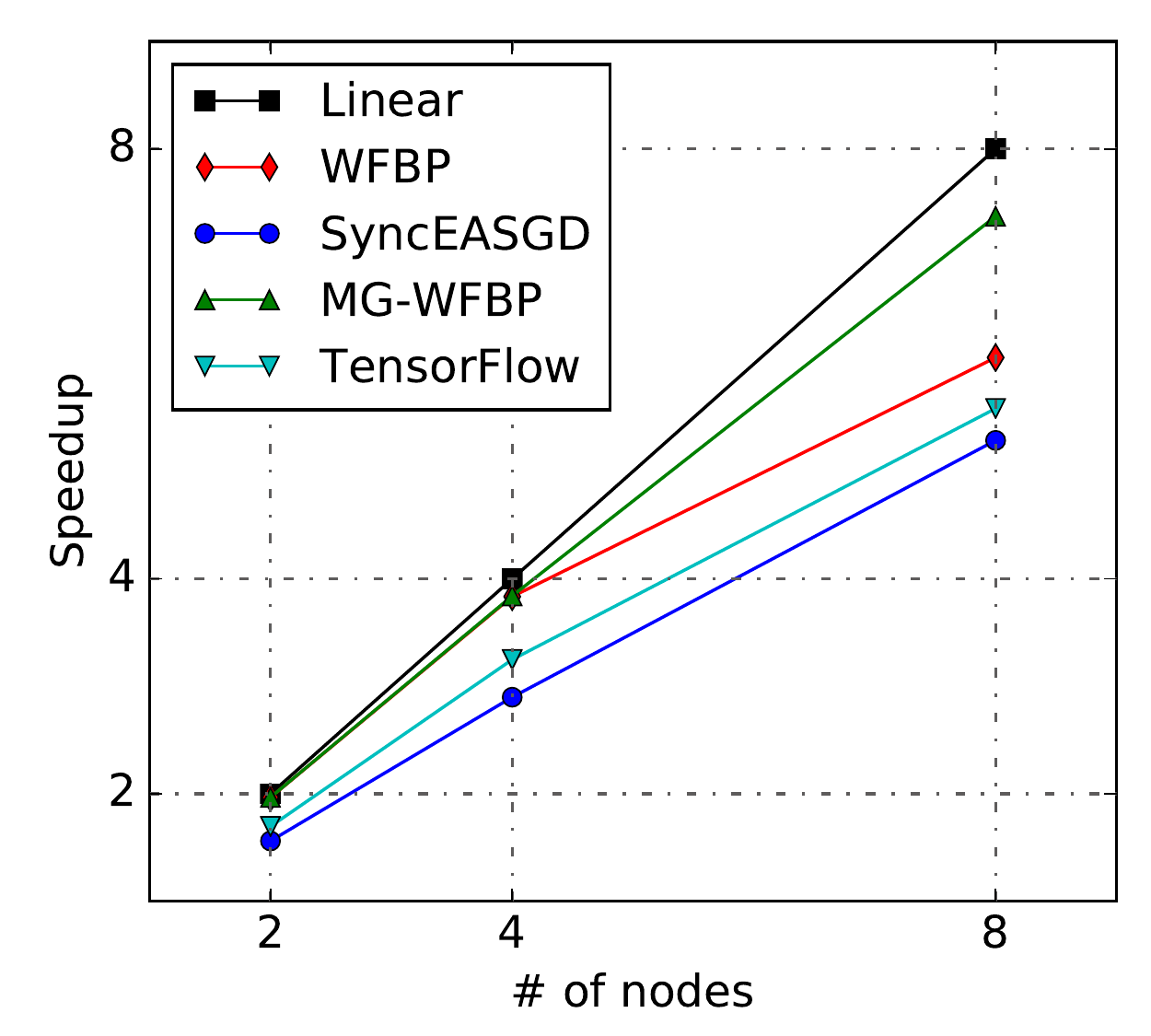}
		\caption{Speedup}
	\end{subfigure}
	\begin{subfigure}{0.24\textwidth}
		\includegraphics[width=\linewidth]{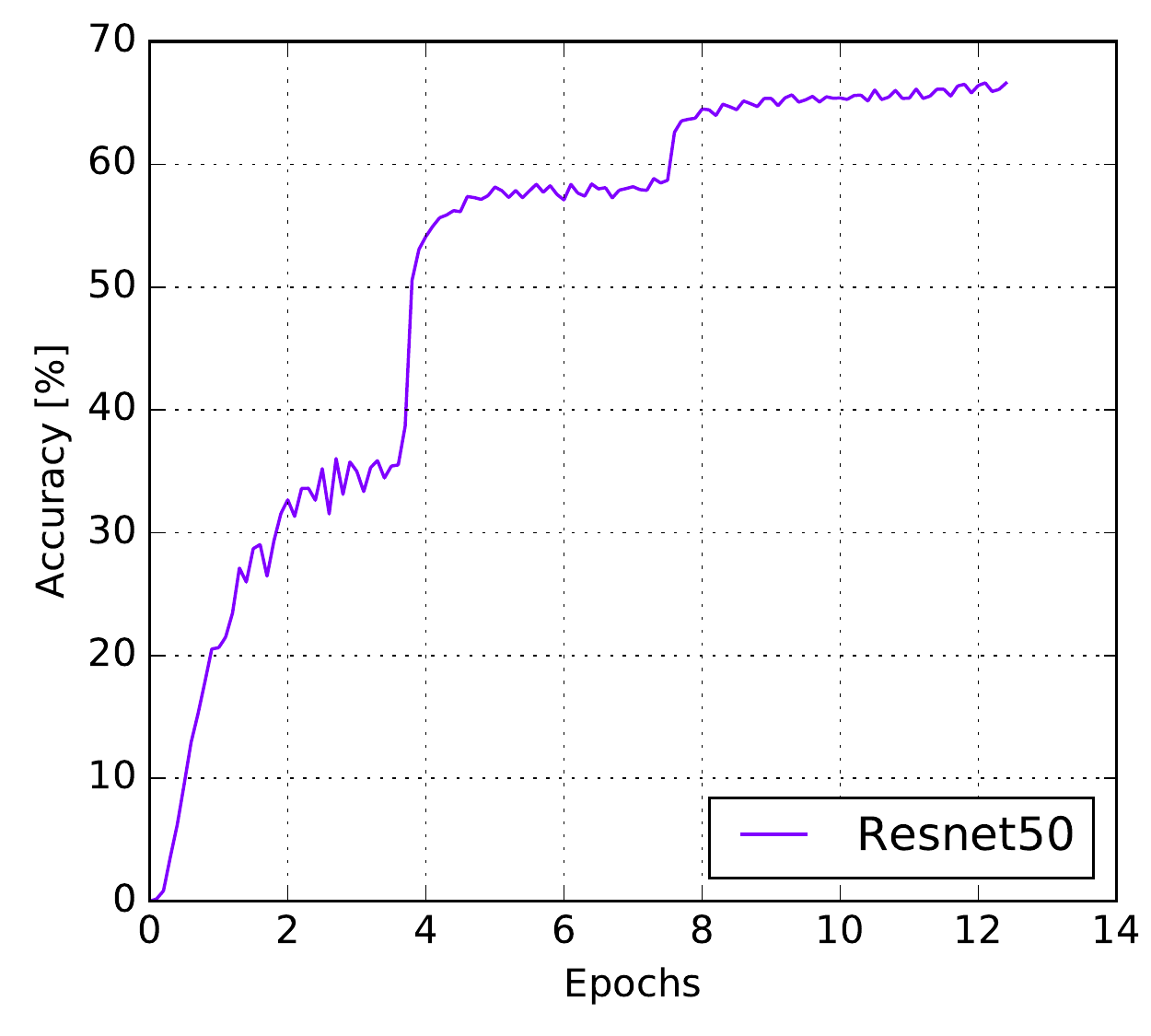}
		\caption{Top-1 validation accuracy}
	\end{subfigure}
	\caption{The performance of ResNet-50 on the K80 cluster with 10GbE. The baseline of the speedup of SGD is on a single machine with 2 GPUs.}
	\label{fig:realresultsresnet}
\end{figure}

\begin{figure}[!h]
	\centering
	\begin{subfigure}{0.24\textwidth}
		\includegraphics[width=\linewidth]{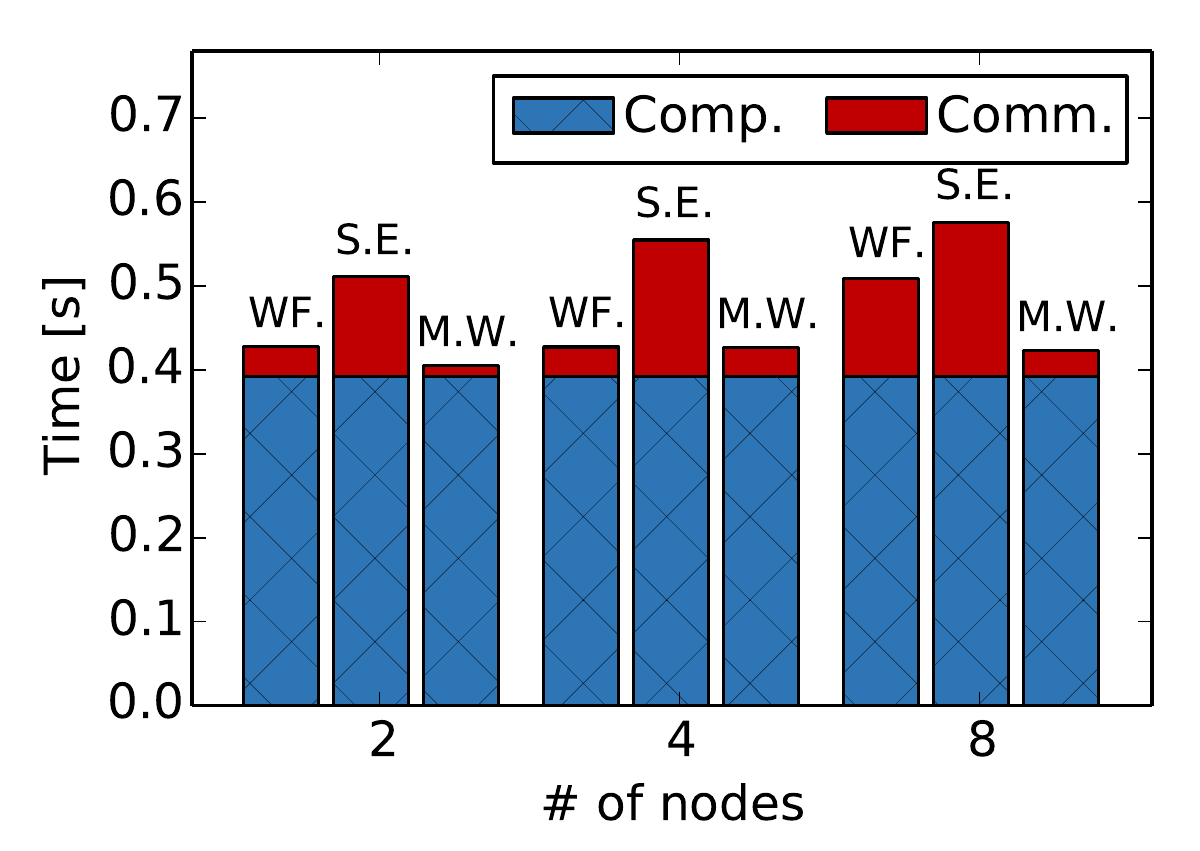}
		\caption{GoogleNet}
	\end{subfigure}
	\begin{subfigure}{0.24\textwidth}
		\includegraphics[width=\linewidth]{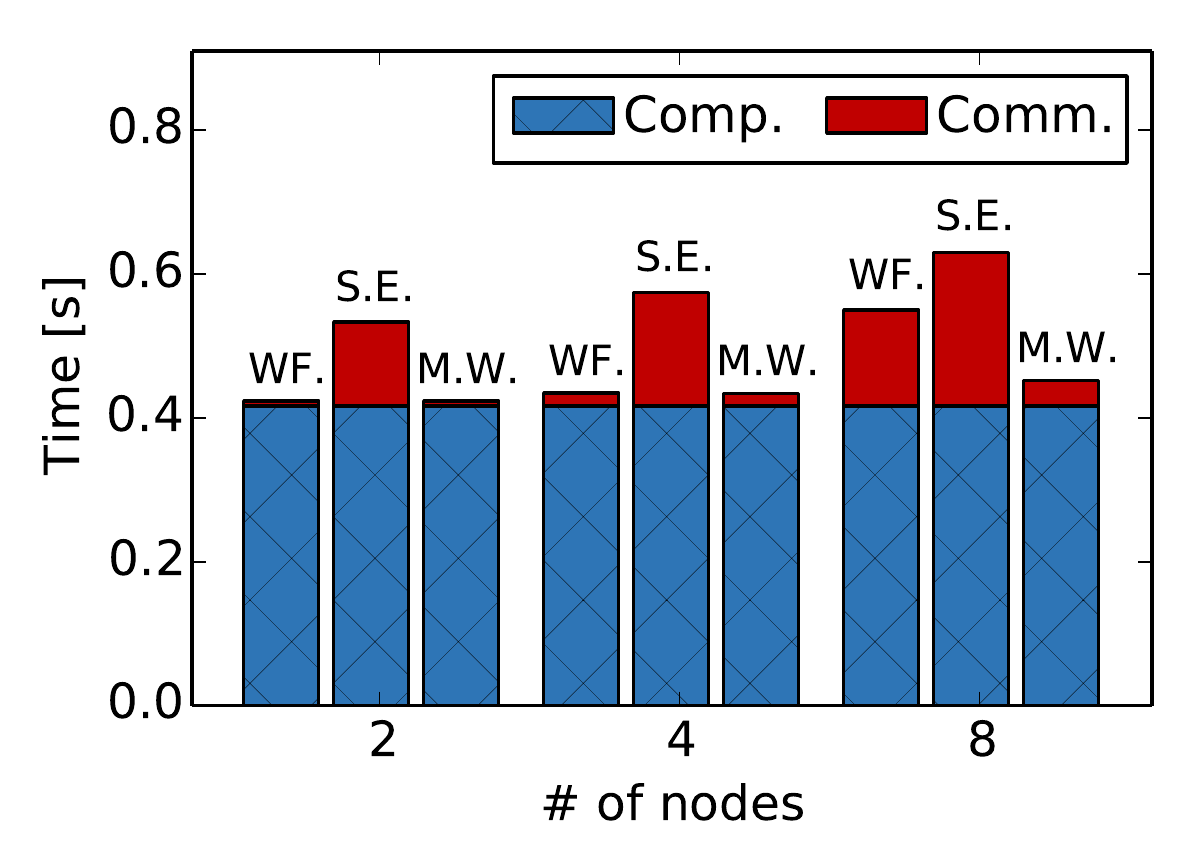}
		\caption{ResNet-50}
	\end{subfigure}
	\caption{Time costs of non-overlapped communication and computation. `WF.', `S.E.' and `M.W.' indicate WFBP, SyncEASGD and MG-WFBP algorithms respectively. `Comp.' refers to the computation cost (i.e., $t_f+t_b$), and `Comm.' refers to the non-overlapped communication cost (i.e., $t_c^{no}$).}
	\label{fig:realcomm}
\end{figure}

The experimental results of GoogleNet and ResNet-50 on the GK210 GPU cluster are shown in Fig. \ref{fig:realresultsgooglenet} and Fig. \ref{fig:realresultsresnet} respectively. The non-overlapped communication cost compared to the computation time is shown in Fig. \ref{fig:realcomm}. The baseline is the iteration throughput of two GPUs in a single machine, in which no communication is required. And the speedup of throughput on multiple workers are compared to the baseline. From Fig. \ref{fig:realcomm}, we can observe that for both GoogleNet and ResNet-50, MG-WFBP performs better than WFBP, SyncEASGD and TensorFlow. SyncEASGD dose not overlap the communication with computation; and hence the communication cost increases when the number of workers increases. As a consequence, the scaling efficiency of SyncEASGD is poor. WFBP achieves near linear scaling on 2 and 4 nodes, in which the non-overlapped communication overhead are small. When scaling to 8 nodes, however, WFBP has an obvious drop in efficiency due to the increased startup time of layer-wise communication which cannot be totally hidden by computation. Regarding the performance of TensorFlow, it uses parameter servers to do the model aggregation. On one hand, the centralized parameter server based algorithm could easily suffer a bandwidth pressure in the parameter server on the lower speed network \cite{zhang2017poseidon}. On the other hand, it takes two communication directions (workers to PS, and PS to workers) to finish the model synchronization, which introduces more overhead in the synchronization pass. Therefore, though TensorFlow exploits the WFBP technique, the PS-based method performs worse than the decentralized method. Our proposed algorithm has a very small non-overlapped communication cost even on the 8-node cluster, so the scaling efficiency is still close to linear. In summary, MG-WFBP achieves about $1.2$x and $1.36$x speedups compared to WFBP and SyncEASGD respectively on the 8-node (16 GPUs) K80 cluster on both GoogleNet and ResNet-50.

\begin{figure*}[!ht]
\captionsetup[subfigure]{justification=centering}
	\centering
	\begin{subfigure}{0.24\textwidth}
		\includegraphics[width=\linewidth]{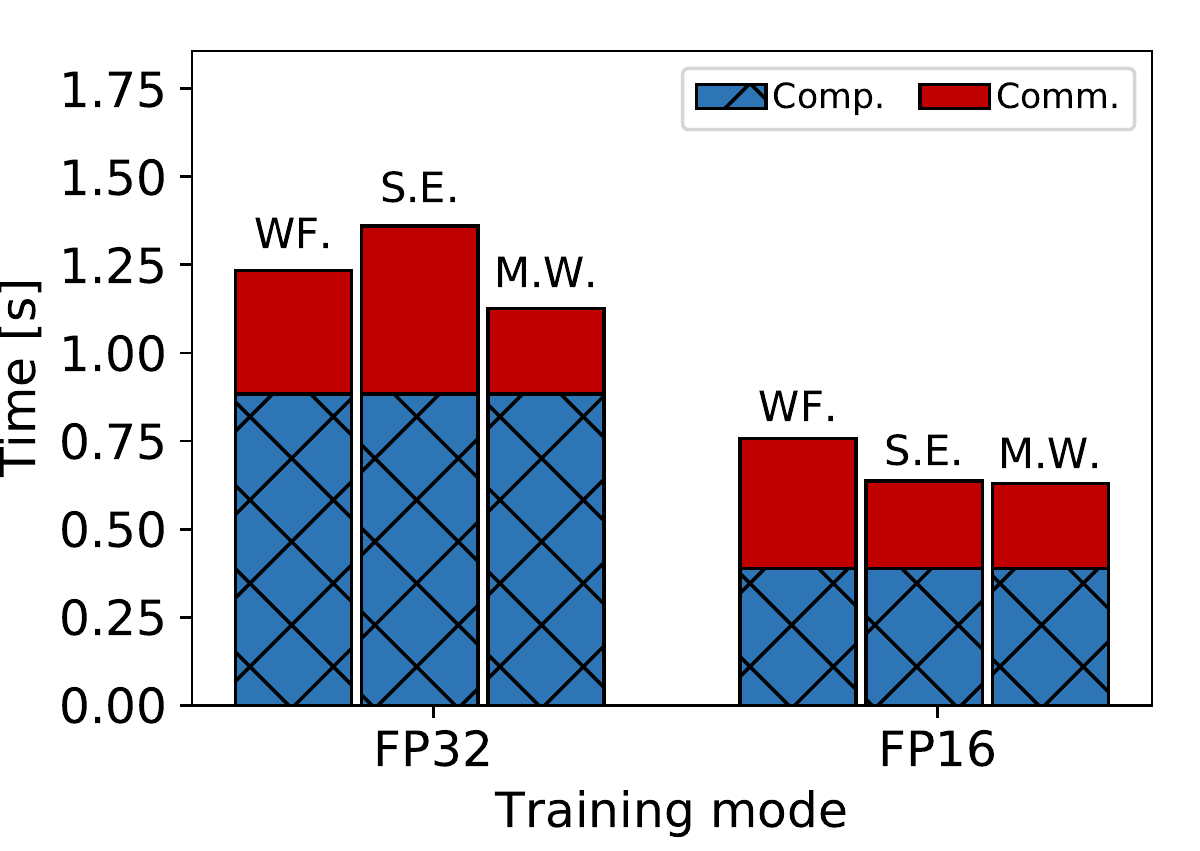}
		\caption{ResNet-152 with 10GbE \\(1\%-20\%)}
	\end{subfigure}
	\begin{subfigure}{0.24\textwidth}
		\includegraphics[width=\linewidth]{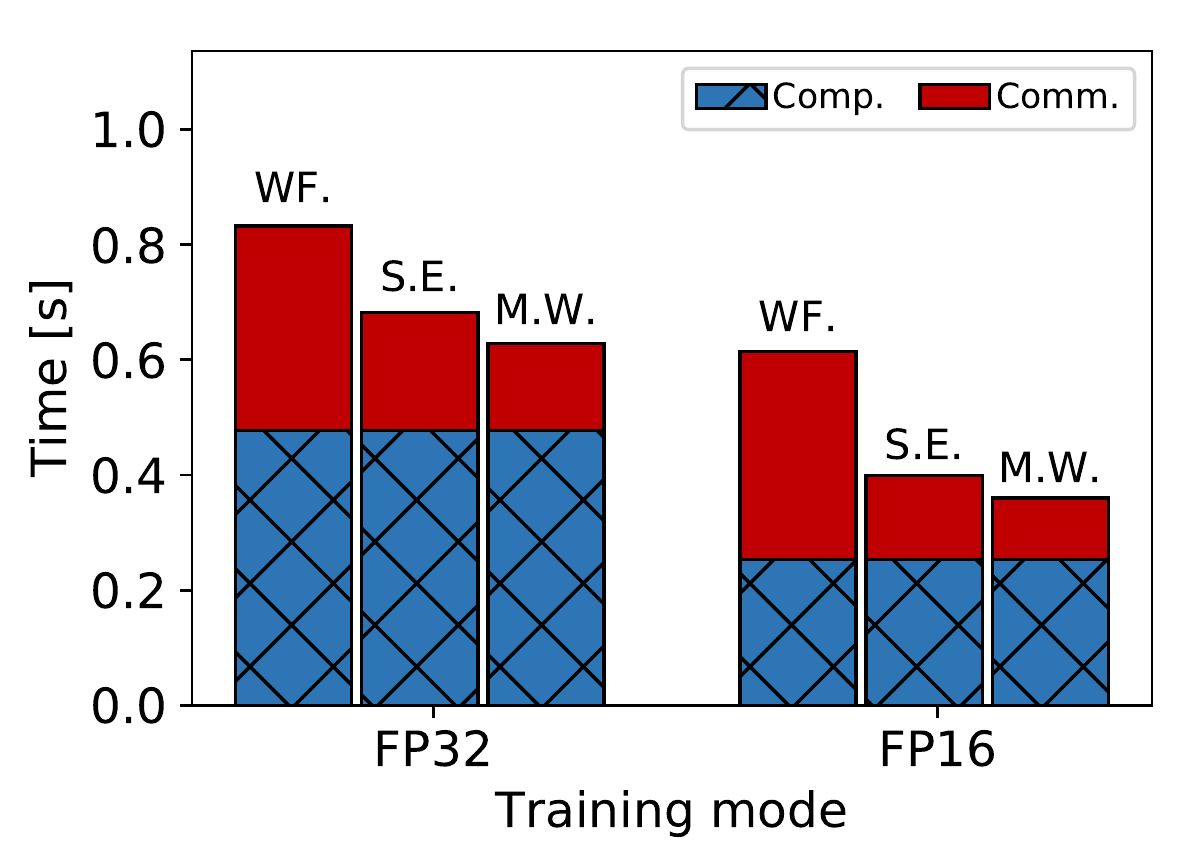}
		\caption{DenseNet-161 with 10GbE \\(7\%-70\%)}
	\end{subfigure}
	\begin{subfigure}{0.24\textwidth}
		\includegraphics[width=\linewidth]{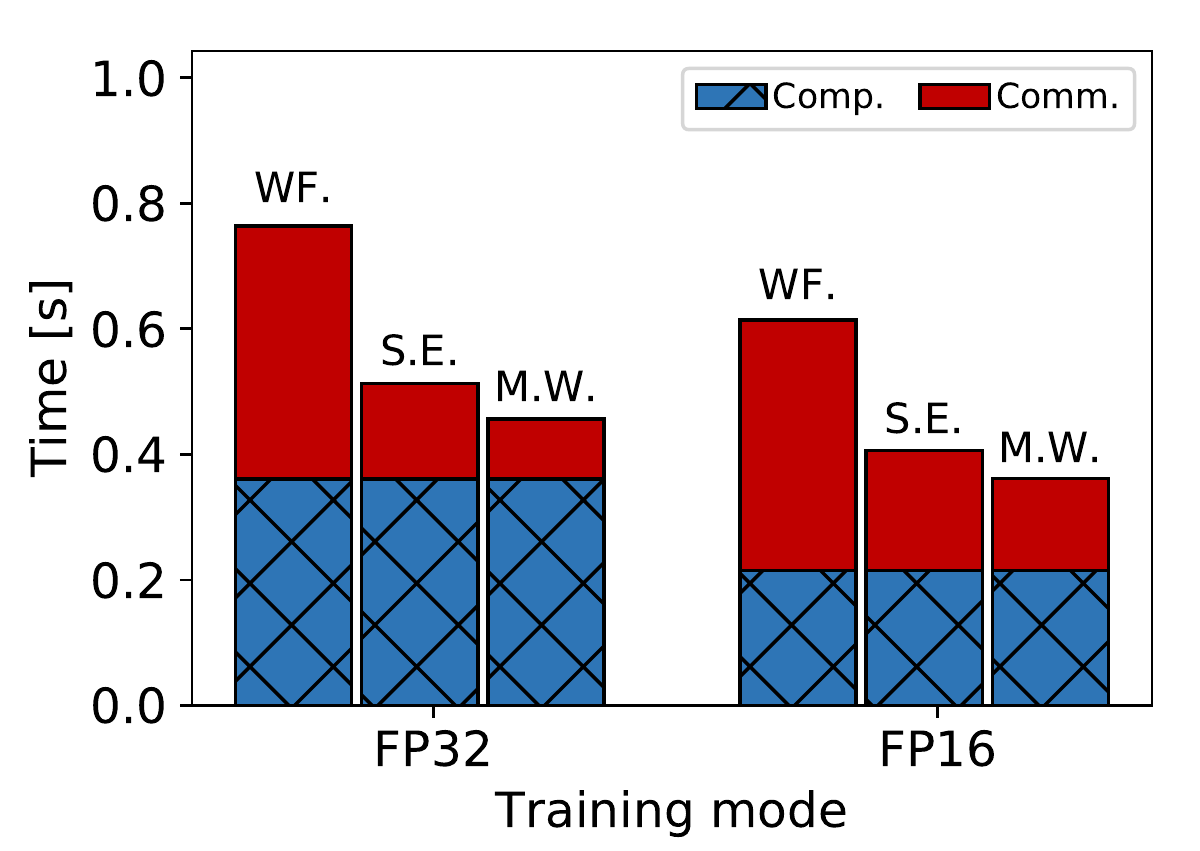}
		\caption{DenseNet-201 with 10GbE \\(7\%-69\%)}
	\end{subfigure}
	\begin{subfigure}{0.24\textwidth}
		\includegraphics[width=\linewidth]{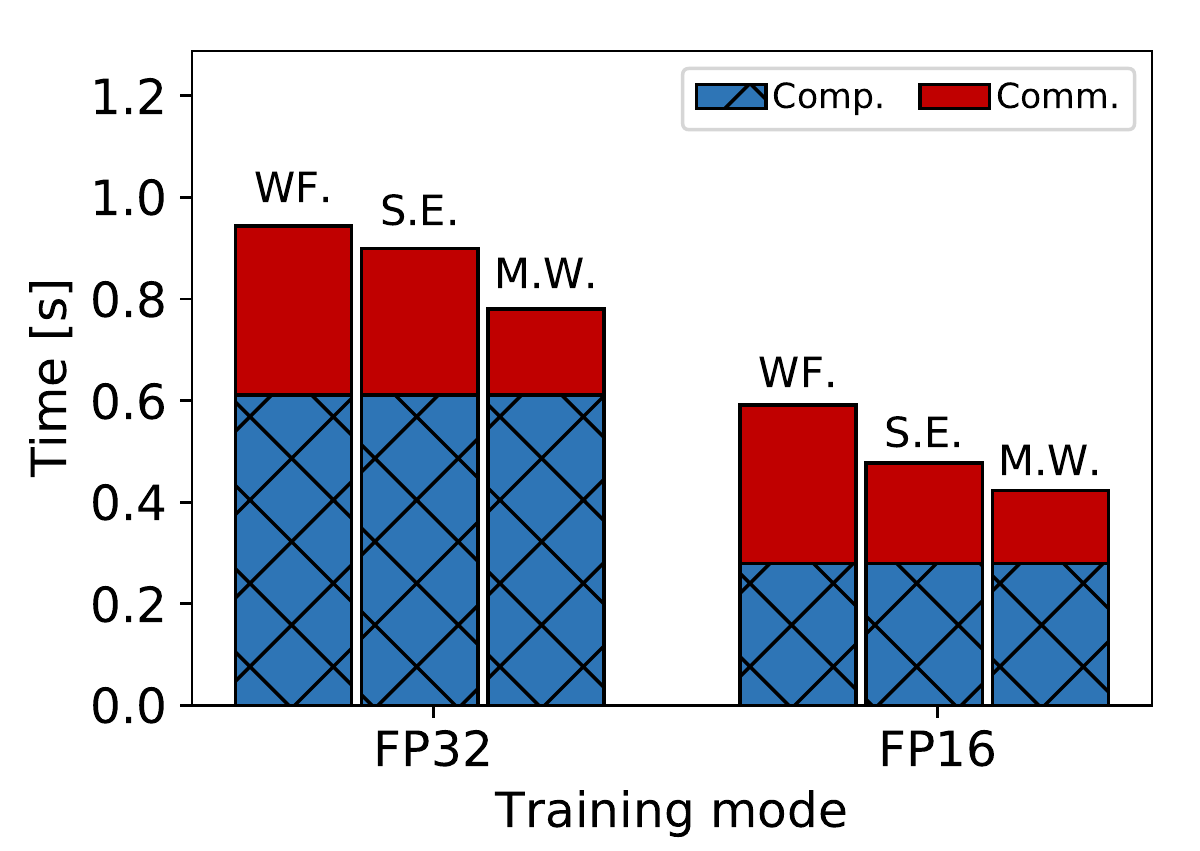}
		\caption{Inception-v4 with 10GbE \\(12\%-39\%)}
	\end{subfigure}
	
	\begin{subfigure}{0.24\textwidth}
		\includegraphics[width=\linewidth]{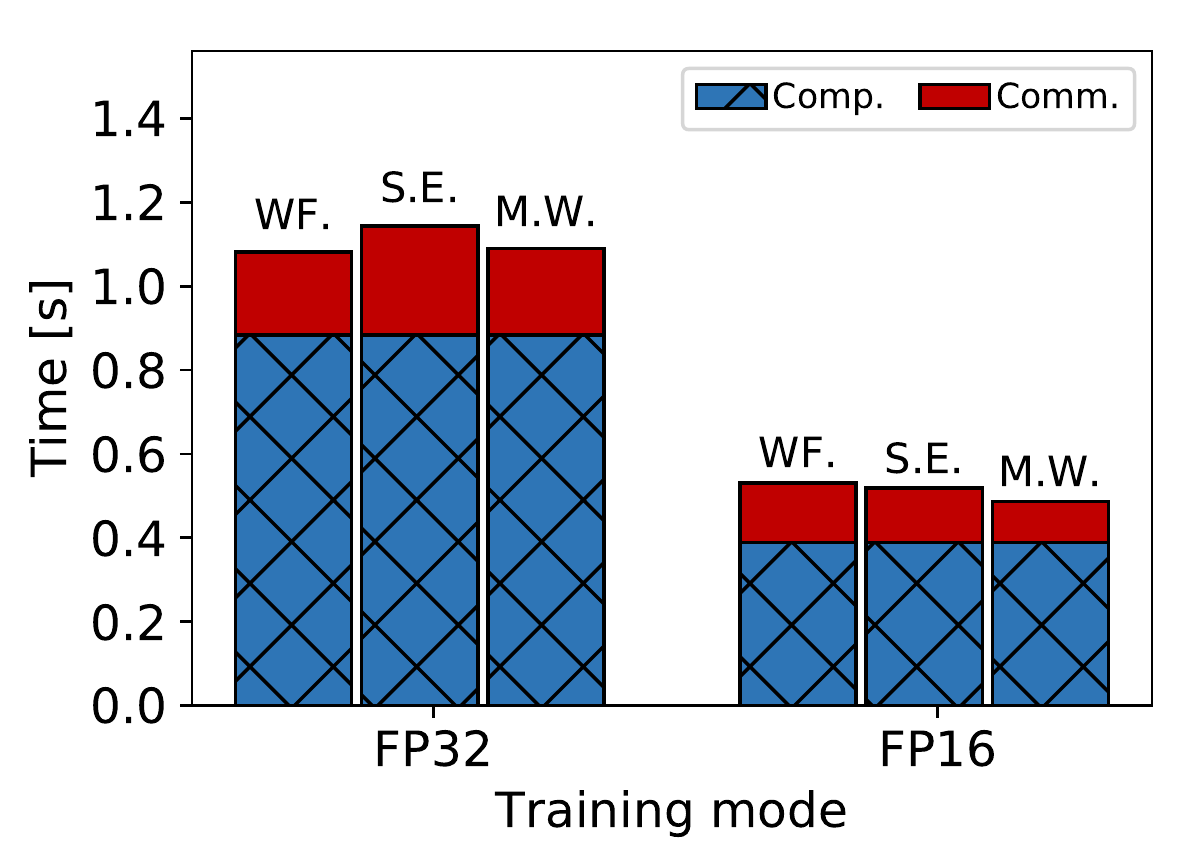}
		\caption{ResNet-152 with 56GbIB \\(2\%-9\%)}
	\end{subfigure}
	\begin{subfigure}{0.24\textwidth}
		\includegraphics[width=\linewidth]{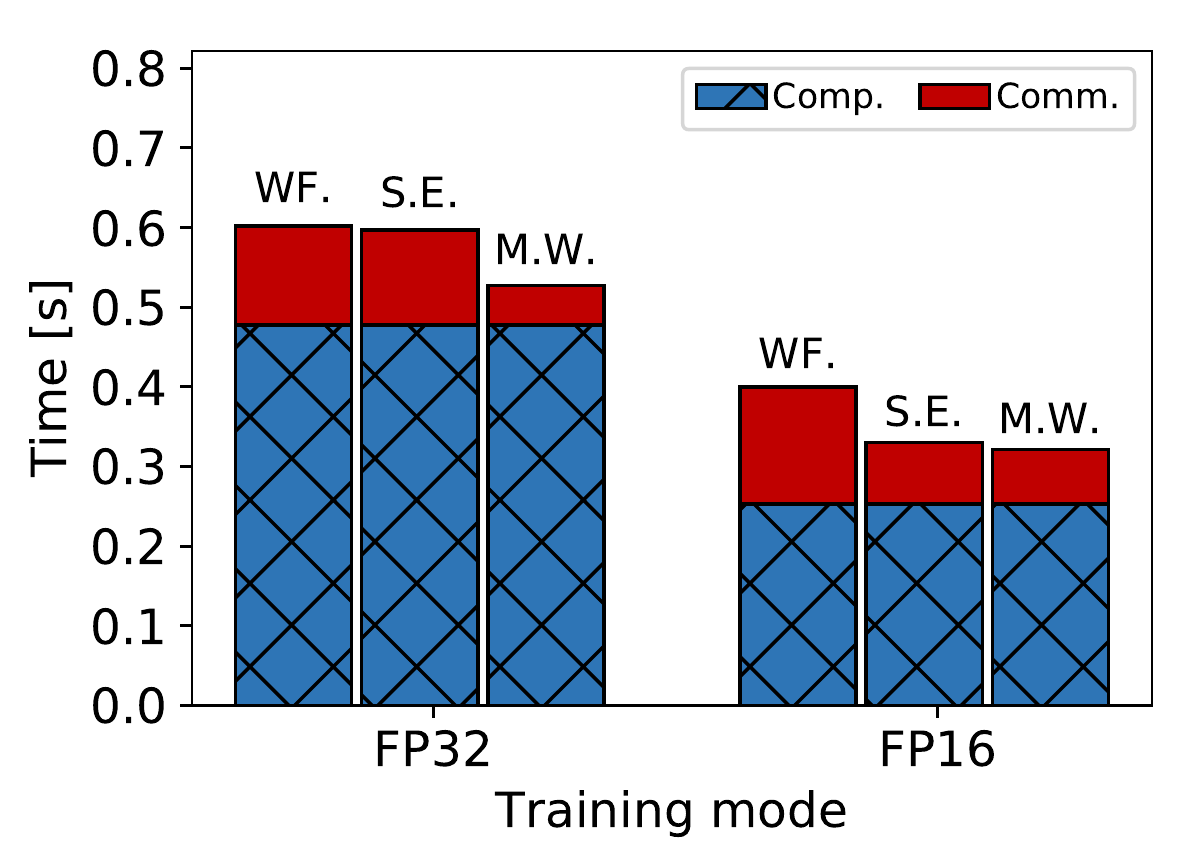}
		\caption{DenseNet-161 with 56GbIB \\(2\%-24\%)}
	\end{subfigure}
	\begin{subfigure}{0.24\textwidth}
		\includegraphics[width=\linewidth]{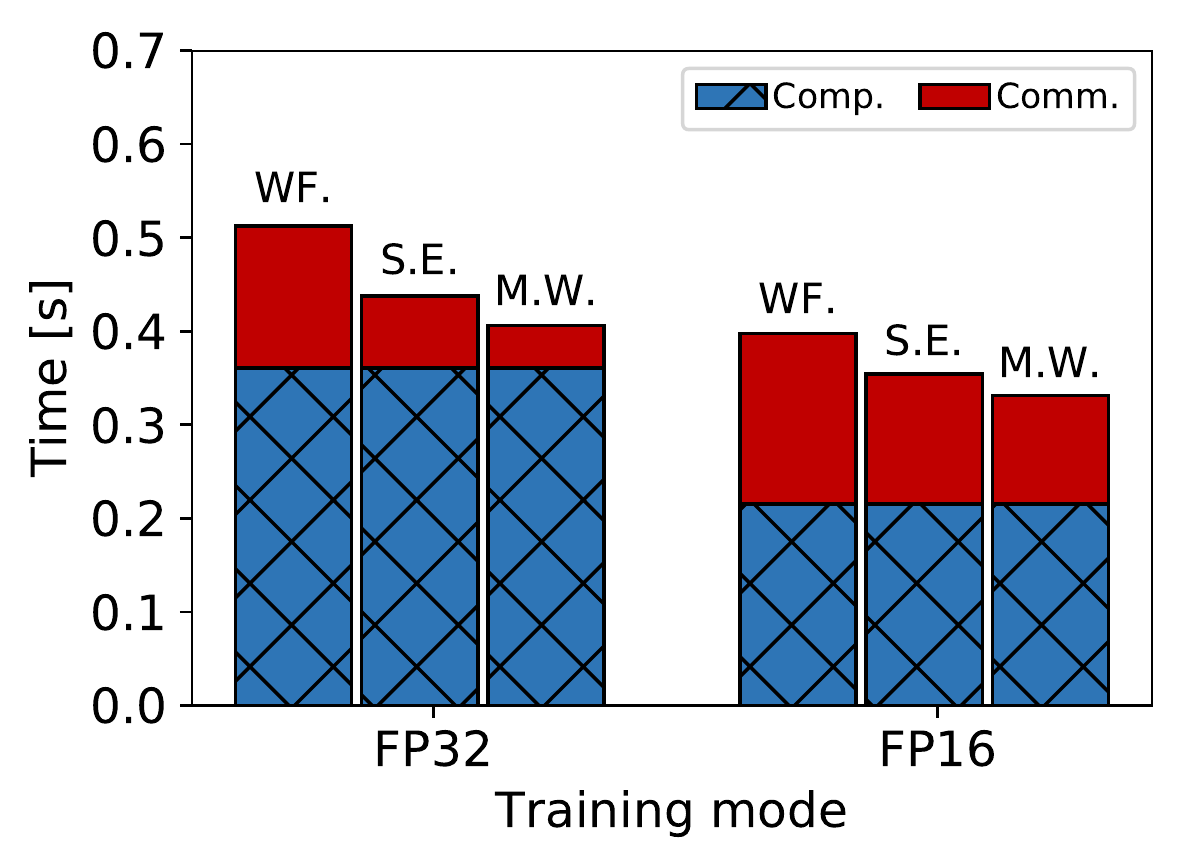}
		\caption{DenseNet-201 with 56GbIB \\(6\%-26\%)}
	\end{subfigure}
	\begin{subfigure}{0.24\textwidth}
		\includegraphics[width=\linewidth]{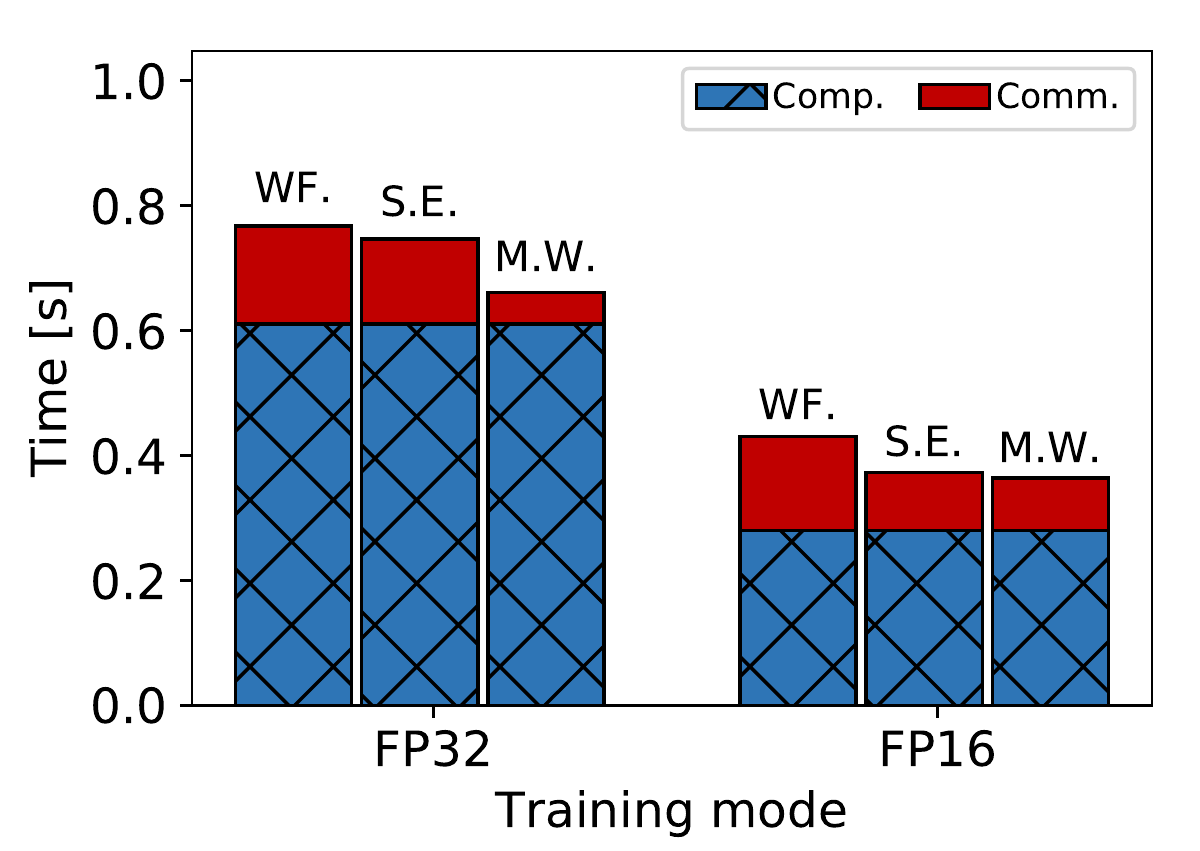}
		\caption{Inception-v4 with 56GbIB \\(2\%-18\%)}
	\end{subfigure}
	\caption{Time comparison of non-overlapped communication and computation on the two V100 GPU clusters (10GbE and 56GbIB). `WF.', `S.E.' and `M.W.' indicate WFBP, SyncEASGD and MG-WFBP algorithms respectively. `Comp.' refers to the computation cost (i.e., $t_f+t_b$), and `Comm.' refers to the non-overlapped communication cost (i.e., $t_c^{no}$). The values in the brackets are the range of improvements of MG-WFBP over WFBP and SyncEASGD.}
	\label{fig:v100results}
\end{figure*}
\subsubsection{Results on Cluster 2 and Cluster 3}
Note that MG-WFBP has no side-effect on the convergence performance (in terms of the number of iterations) as MG-WFBP can achieve consistent results of the aggregated gradients with the original S-SGD at each iteration. Therefore, in the following performance evaluation, we focus on the comparison on the average iteration wall-clock time to demonstrate how much performance improvement of our MG-WFBP over WFBP and SyncEASGD. 

On cluster 2 and cluster 3, in addition to the general setting with single precision (FP32) training, we also apply our MG-WFBP algorithm to the mixed precision training technique \cite{micikevicius2018mixed}, which is widely used on the GPUs with Tensor Cores (e.g., Tesla V100) to increase the computing efficiency and reduce the communication traffic. The results are shown in Fig. \ref{fig:v100results}. In overall, it can be seen that for different DNN models, no one always outperforms the other one between WFBP and SyncEASGD algorithms as the both algorithms are sensitive to the cluster configurations, while our proposed MG-WFBP algorithm achieves the fastest training speed in all evaluated DNNs. The first row of Fig. \ref{fig:v100results} shows that MG-WFBP achieves up to $70\%$ improvement over WFBP and SyncEASGD algorithms on Cluster 2 with 10GbE connection. The second row of Fig. \ref{fig:v100results} demonstrates that MG-WFBP outperforms WFBP and SyncEASGD up to $26\%$ on Cluster 3 with 56GbIB connection.

On the ResNet-152 architecture, pipelining all FP32 tensors brings some benefits to hide some communication overheads so that WFBP trains faster than SyncEASGD. On both DenseNet and Inception architectures, however, pipelining for every tensors between communication and computation introduces many extra communication overheads so that WFBP performs slower training speed than SyncEASGD. On the ResNet-152 architecture with FP32 precision, the hidden communication time is longer than the extra time introduced by each layer's startup overhead with pipelining so that WFBP is about $10\%$ faster than SyncEASGD. Our MG-WFBP algorithm can further reduce the negative impact of the startup time by smartly merging some gradients, which results in extra $10\%$ improvement. On the other hand, pipelining all tensors introduces larger overheads than hidden time. For example, SyncEASGD is $20\%$ faster than WFBP in DenseNet-161. By merging the tensors smartly, MG-WFBP performs $7\%$ faster than SyncEASGD. 

In summary, MG-WFBP can always outperform WFBP and SyncEASGD. In the conducted extensive experiments, MG-WFBP generally achieves up to $15\%$ improvement over the best of WFBP and SyncEASGD in both 10GbE and 56GbIB interconnections.

\subsection{Simulation}
\begin{figure}[!ht]
	\centering
    \begin{subfigure}{0.24\textwidth}
		\includegraphics[width=\linewidth]{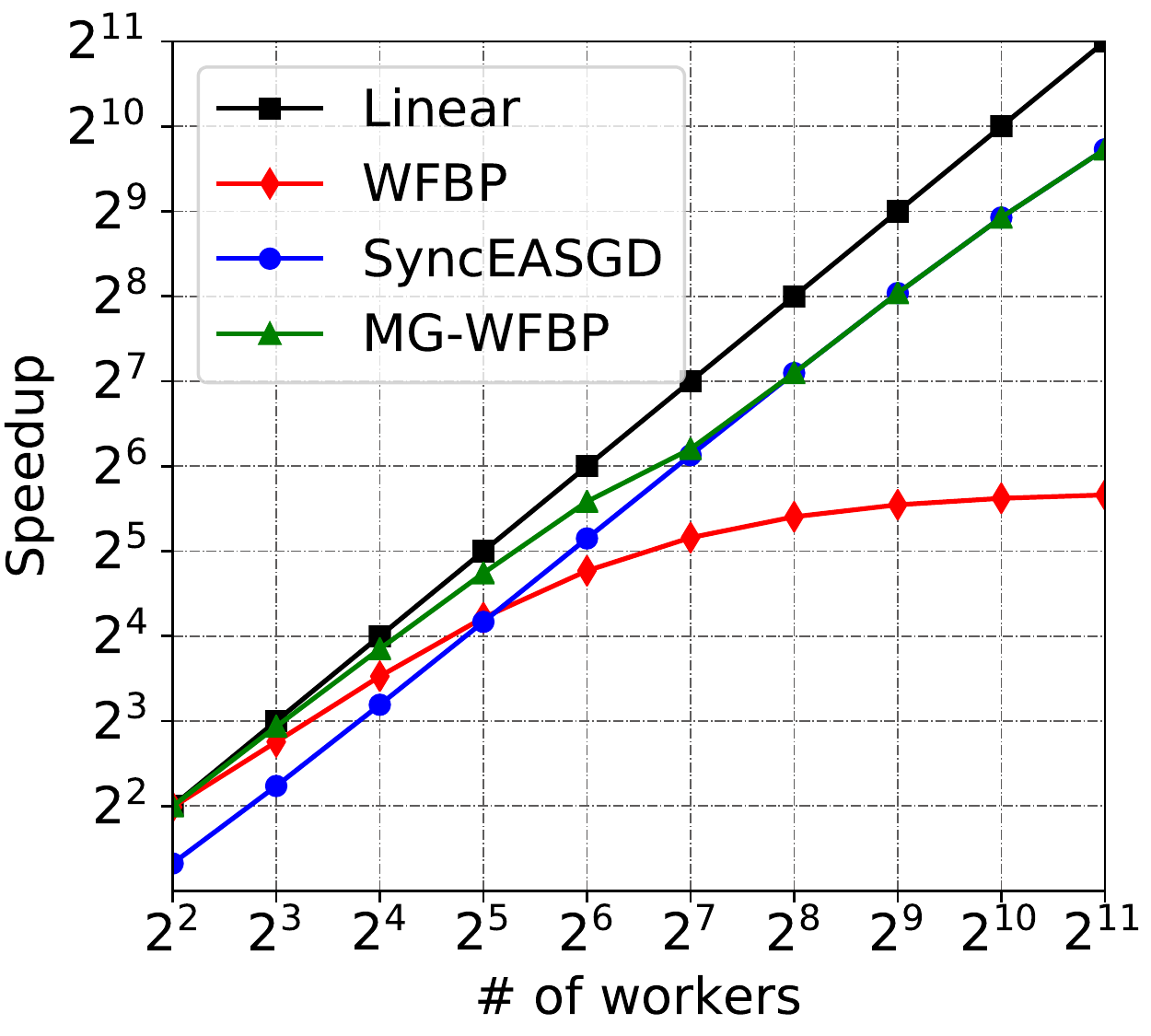}
		\caption{GoogleNet}
	\end{subfigure}
    \begin{subfigure}{0.24\textwidth}
		\includegraphics[width=\linewidth]{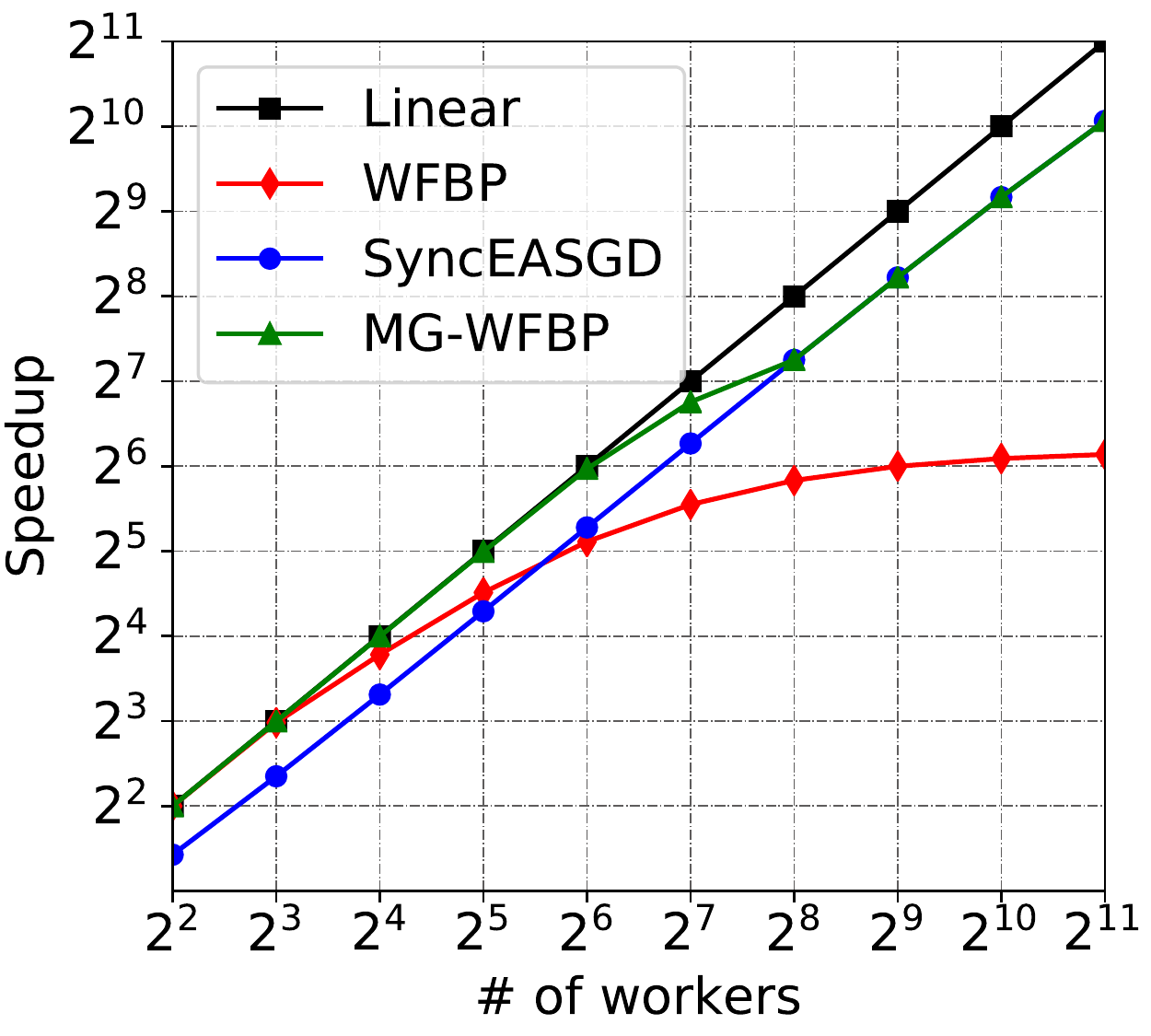}
		\caption{ResNet-50}
	\end{subfigure}
	\caption{The performance comparison on the simulated K80 cluster connected with 10GbE with the ring-based all-reduce algorithm. Baseline of the speedup of SGD is on a single K80 card.}
	\label{fig:simspeedupk80}
\end{figure}
Due to the hardware limitation, we do not have a very large GPU cluster to support more large-scale experiments. So we conduct simulations based on the real single-GPU performance and the network performance model. First, we measure the layer-wise backward propagation time (i.e., the computation time $t_b^{(l)}$, $l=1,2,...,L$) of GoogleNet and ResNet-50 on a single K80 GPU. Second, to estimate the parameters of $a$ and $b$ in the communication model of Eq.~\ref{equ:tcomm}, we exploit the fitted parameters shown in Fig.~\ref{fig:commoverhead} on a K80 GPU cluster connected with 10GbE. Based on the measured layer-wise backward propagation time on the real K80 GPU and the communication model on the K80 GPU cluster, we simulate WFBP, SyncEASGD and MG-WFBP by scaling from 4 workers to 2048 workers with the ring-based and double binary tree based all-reduce algorithms, which have been practically implemented by NCCL.

\textbf{Overall Performance with ring-based all-reduce}. We simulate to train GoogleNet and ResNet-50 by scaling from 4 workers to 2048 workers. The scaling performance is shown in Fig. \ref{fig:simspeedupk80}, in which MG-WFBP has $n=[10, 6, 6, 5, 3, 2, 1,..., 1]$ and $n=[33,19,10,7,5,3,1,...,1]$ merged-gradient layers in GoogleNet and ResNet-50 respectively on the $p=[2^2, 2^3, ..., 2^{11}]$-worker clusters. On the 64-worker cluster, MG-WFBP outperforms WFBP and SyncEASGD by $1.78$x and $1.35$x, respectively on GoogleNet. On ResNet-50, MG-WFBP achieves almost linear speedup, while WFBP and SyncEASGD only have around $55\%$ scaling efficiency. It is important to notice that the curves of WFBP and SyncEASGD have a crossing point in Fig. \ref{fig:simspeedupk80}. This is because the two algorithms are sub-optimal in utilizing the network bandwidth; when the startup time of network communication is not very large (e.g., 4-16 workers in the K80 cluster), WFBP would have the advantage to hide more communication time compared to SyncEASGD. But when scaling to medium-size clusters (e.g., 64 workers), the startup time of communication becomes much larger so that it is hard to be hidden, then using a single-layer communication could become a better approach. As we can see, SyncEASGD achieves better scaling efficiency than WFBP in the 64-worker cluster on both tested CNNs. In such scenarios, MG-WFBP not only overlaps the communication with computation, but also finds the optimal communication message size. So it achieves better scaling efficiency than SyncEASGD and WFBP. Similarly, on training ResNet-50, MG-WFBP achieves about $1.75$x and $1.45$x speedups compared to WFBP and SyncEASGD respectively on the simulated 64-worker cluster. When scaling to large-size clusters (e.g., 256 workers or more), our MG-WFBP converges to the SyncEASGD since the startup time of each layer becomes too large to be hidden, which means that the single-layer communication becomes the optimal. In summary, on the simulated experiments, our proposed algorithm MG-WFBP always achieves the best speedup. However, the ring-based all-reduce algorithm has a startup time that is linear to the number of workers, which makes MG-WFBP become the single-layer communication when scaling to large-scale clusters. 

\textbf{Simulation with double binary trees}. The startup term in ring-based all-reduce is linear to the number of workers, hence it does not perform well in very large clusters. In the recent NCCL releases (from version 2.4), the double binary trees all-reduce algorithm~\cite{sanders2009two} becomes an alternative as it has a logarithmic startup overhead. We replace $a$ and $b$ with the double binary trees algorithm as shown in Table~\ref{table:allreduce} to compare SyncEASGD, WFBP, and MG-WFBP with simulations on 128 to 2048 workers. The results are shown in Fig.~\ref{fig:simspeedupk80-tree}. It can be seen that WFBP and MG-WFBP always outperform SyncEASGD as the startup time of the double binary trees algorithm is relatively small. With the gradient merge solution, MG-WFBP achieves better performance than WFBP by eliminating some layer's startup times. 
\begin{figure}[!ht]
	\centering
    \begin{subfigure}{0.24\textwidth}
		\includegraphics[width=\linewidth]{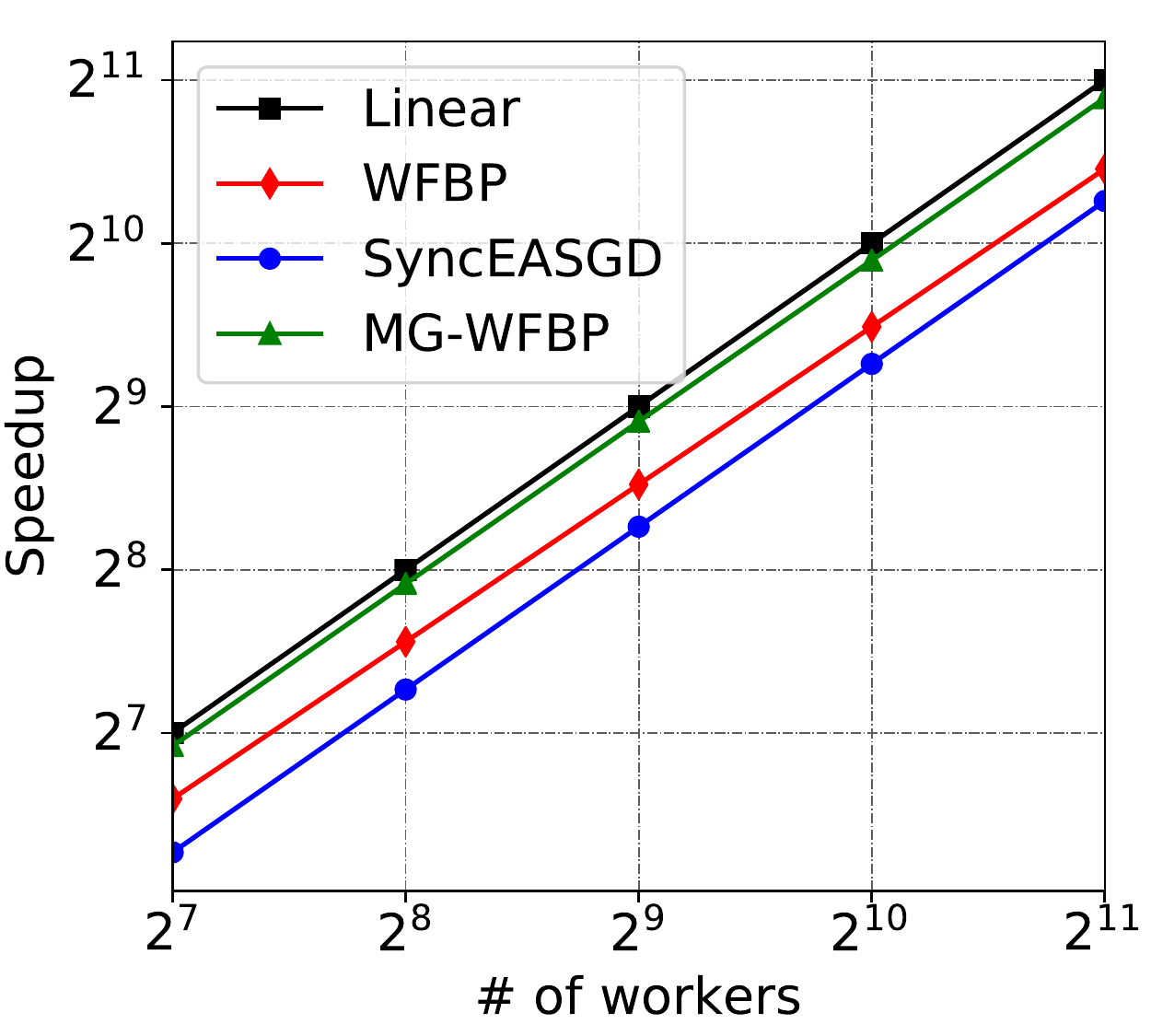}
		\caption{GoogleNet}
	\end{subfigure}
    \begin{subfigure}{0.24\textwidth}
		\includegraphics[width=\linewidth]{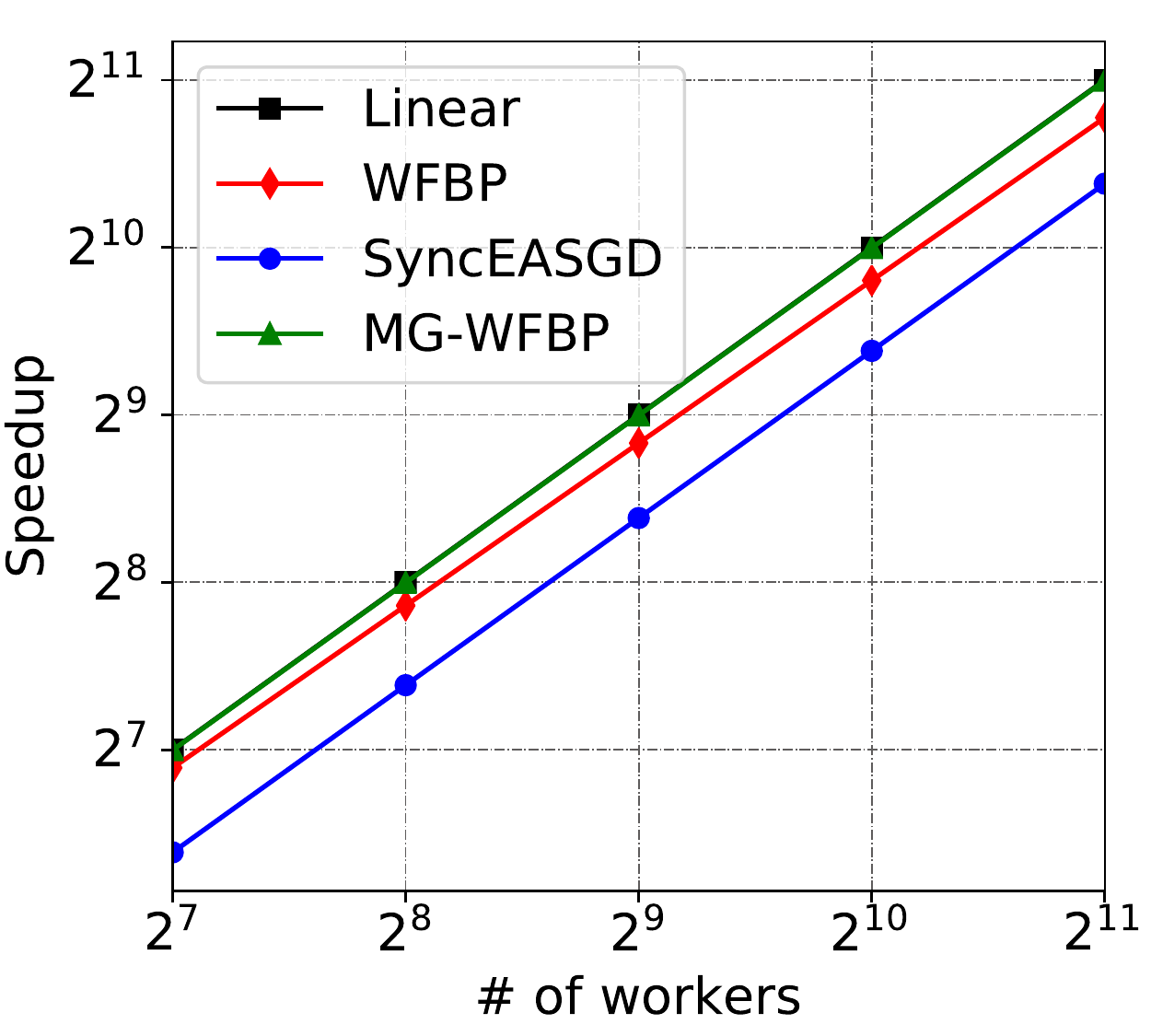}
		\caption{ResNet-50}
	\end{subfigure}
	\caption{The performance comparison on the simulated K80 cluster connected with 10GbE with the double binary trees all-reduce algorithm. }
	\label{fig:simspeedupk80-tree}
\end{figure}

\section{Related Work}\label{s:relatedwork}
The wait-free backward propagation (WFBP) algorithm has recently been proposed to reduce such impact by overlapping communication with computation \cite{awan2017s}\cite{zhang2017poseidon}. In WFBP, the backward computation operations can be started without waiting for the completion of the previous round of data communication. If the communication cost of layer $l+1$ is smaller than the cost of gradients computation of layer $l$, then the communication cost can be completely hidden (except the first layer); and as a result, the scaling efficiency can be close to linear \cite{awan2017s}\cite{zhang2017poseidon}. In practice, however, many DNN models are trained on high-throughput GPUs that result in very short computing time for each backward layer, while it needs to wait for gradient aggregation before starting the next iteration especially on low bandwidth networks (e.g., 10 Gbps Ethernet). 

Current distributed training systems \cite{you2017scaling}\cite{hoefler2010toward}\cite{jia2018highly} exploit tensor fusion that merges small size of gradients before communicating across workers to reduce the communication overhead. The parameter server (PS) method \cite{li2014communication} is proposed for parallelism between computation and communication, but it easily suffers from the communication traffic jam since PS needs to collect the gradients from all the workers. In the centralized framework, Pumma et al. \cite{pumma2017parallel}\cite{pumma2019scalable} provided detailed analysis on the data I/O bottleneck and optimization for large-scale training. Sufficient factor broadcasting (SFB) \cite{zhang2017poseidon} uses the matrix factorization technique to reduce the volume of the data that needs to be communicated for fully connected layers. Although SFB uses P2P communication to eliminate the bandwidth pressure on the PS, it brings a growing number of sufficient factors with both the increasing number of data samples and workers. Zhang et al. \cite{zhang2017poseidon} proposed the Poseidon system with hybrid communication of PS and SFB combined with the WFBP algorithm, and they have achieved 15.5x speedup on 16 single-GPU (TITAN X Pascal) machines. Unfortunately, due to drawbacks of PS and SFB and the communication scheme, Poseidon could also be far away from linear scaling with a large number of workers. 

In the HPC community, the MPI data communication collectives have been redesigned for distributed training to improve the communication performance across multiple machines \cite{awan2017s}. Many MPI-like implementations, such as OpenMPI, NCCL, Gloo\footnote{https://github.com/facebookincubator/gloo} and MVAPICH2-GDR\footnote{https://mvapich.cse.ohio-state.edu/}, support efficient CUDA-aware communication between GPUs via network, and many state-of-the-art DL frameworks (e.g., TensorFlow, PyTorch, Caffe2 and CNTK\footnote{\url{https://docs.microsoft.com/en-us/cognitive-toolkit/}}) integrate NCCL or Gloo for their distributed training modules. Even though these libraries provide very efficient communication collectives, the data communication would still become bottleneck when the communication-to-computation ratio is high, and S-SGD does not scale very well. 

\section{Discussion}\label{s:discission}
Our proposed MG-WFBP is an efficient solution to alleviate the impact of the startup overhead of network communications in distributed DL, but it still has the following limitations: 1) it assumes synchronized SGD with data parallelism, and 2) it requires extra GPU memory (with the same size as model parameters) to buffer the gradients during training. 

The MG-WFBP algorithm mainly considers the scheduling in the backward pass of S-SGD. It would be possible to extend MG-WFBP to a more general scheduling framework. For example, in S-SGD, gradient compression~\cite{lin2018deep,shi2019adistributed} is a promising approach to improving the scalability of distributed DL~\cite{tang2020communication}. To integrate gradient compression with MG-WFBP, one should consider the extra computational overhead of gradient compression (e.g., top-k sparsification~\cite{lin2018deep}) to generate an optimal solution~\cite{shi2020communication}. Furthermore, it is also possible to pipeline the communications and feed-forward computations so that some communication overheads can be hidden during the feed-forward pass~\cite{bao2020preemptive}. It could be more challenging and useful by considering both feed-forward and backward passes to achieve an optimal gradient merge solution. We will leave this as our future work.

\section{Conclusion}\label{s:conclusion}
In this work, we first showed that existing state-of-the-art communication strategies, say wait-free backward propagation (WFBP) and single-layer communication (SyncEASGD), are sub-optimal in the synchronized distributed deep learning training when the communication-to-computation ratio is high. Then we generalized the communication problem in pipelining communication and computation as an optimization problem and developed an optimal solution with an efficient algorithm. We then proposed the merged-gradient wait-free backward propagation (MG-WFBP) strategy by optimally merging gradients. We implemented MG-WFBP atop the popular deep learning framework PyTorch. Our implementation is also publicly available. Through extensive experiments on three 16-GPU clusters including Nvidia Tesla K80 GPUs with 10Gbps Ethernet connection and Nvidia Tesla V100 GPUs with both 10Gbps Ethernet and 56Gbps InfiniBand, we verified that MG-WFBP can achieve much better scalability than WFBP and SyncEASGD on various popular convolutional neural networks. Simulations were also studied to further explore the advantage of MG-WFBP on large-scale clusters.

\section*{Acknowledgments}
The research was supported in part by Hong Kong RGC GRF grants under the contracts HKBU 12200418, HKUST 16206417 and 16207818, as well as an RGC CRF grant under the contract C7036-15G. We would also like to thank NVIDIA for providing the GPU clusters for experiments.

\bibliographystyle{IEEEtran}
\Urlmuskip=0mu plus 1mu
\bibliography{merged_gradients_tpds.bbl}
\vspace{-20pt}
\begin{IEEEbiography}[{\includegraphics[width=1in,height=1.25in,clip,keepaspectratio]{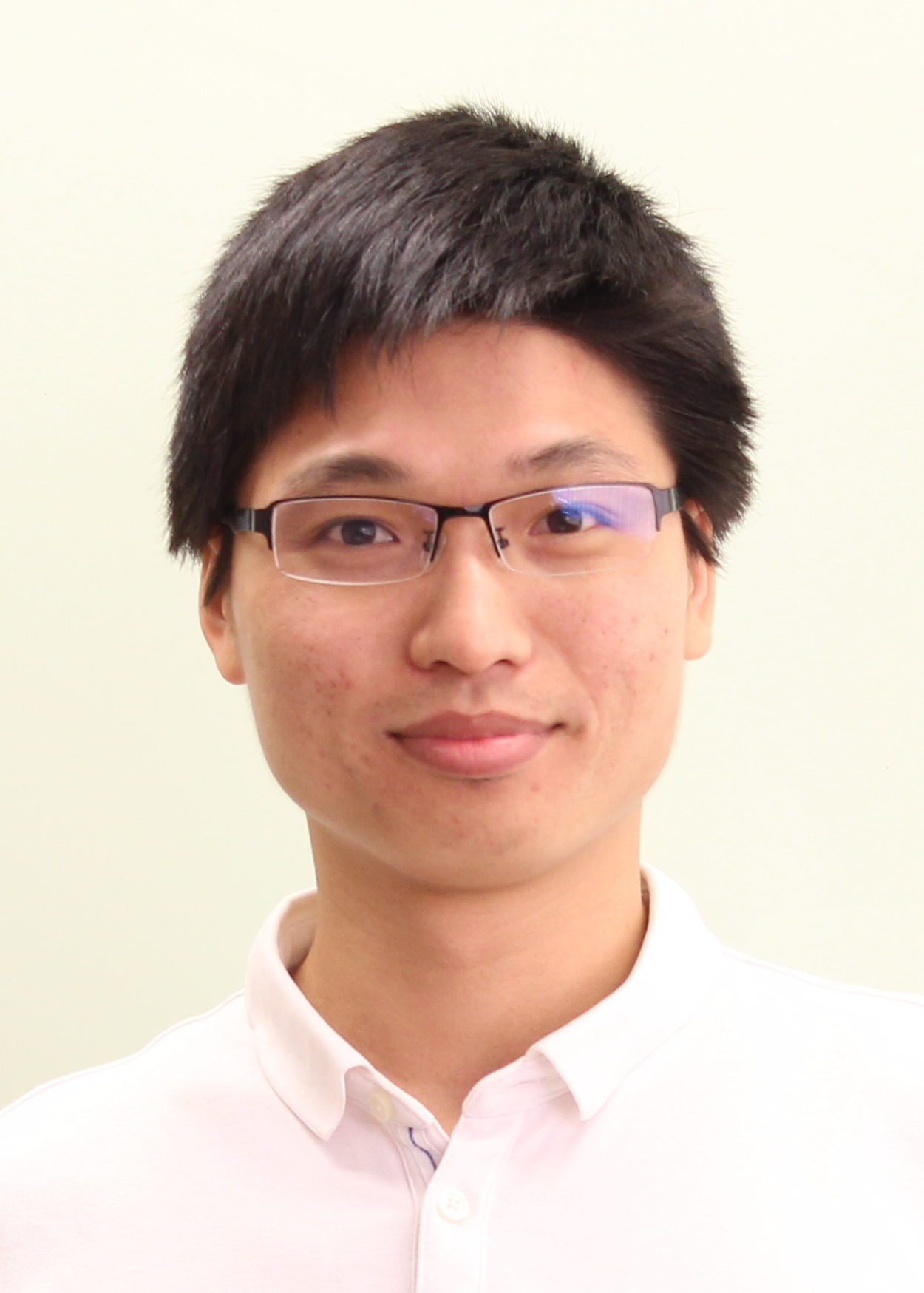}}]{Shaohuai Shi} received a B.E. degree in software engineering from South China University of Technology, P.R. China, in 2010, an MS degree in computer science from Harbin Institute of Technology, P.R. China in 2013, and a Ph.D. degree in computer science from Hong Kong Baptist University in 2020. He is currently a research assistant professor in the Department of Computer Science and Engineering at the Hong Kong University of Science and Technology. His research interests include GPU computing and machine learning systems. He is a member of the IEEE.
\end{IEEEbiography}
\vspace{-35pt}
\begin{IEEEbiography}[{\includegraphics[width=1in,height=1.25in,clip,keepaspectratio]{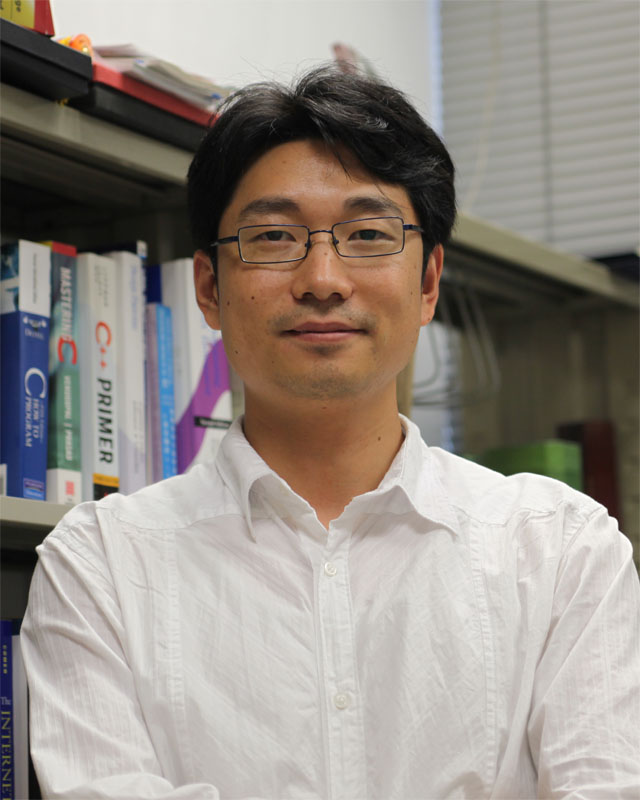}}]{Xiaowen Chu} received the B.E. degree in computer science from Tsinghua University, P.R. China, in 1999, and the Ph.D. degree in computer science from The Hong Kong University of Science and Technology in 2003. Currently, he is a full professor in the Department of Computer Science, Hong Kong Baptist University. His research interests include parallel and distributed computing, cloud computing and wireless networks. He is serving as an Associate Editor of IEEE Access and IEEE Internet of Things Journal. He is a senior member of the IEEE.
\end{IEEEbiography}
\vspace{-35pt}
\begin{IEEEbiography}[{\includegraphics[width=1in,height=1.25in,clip,keepaspectratio]{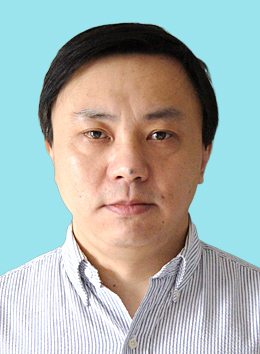}}]{Bo Li} is a professor in the Department of Computer Science and Engineering, Hong Kong University of Science and Technology. He holds the Cheung Kong chair professor in Shanghai Jiao Tong University. Prior to that, he was with IBM Networking System Division, Research Triangle Park, North Carolina. He was an adjunct researcher with Microsoft Research Asia-MSRA and was a visiting scientist in Microsoft Advanced Technology Center (ATC). He has been a technical advisor for China Cache Corp. (NASDAQ CCIH) since 2007. He is an adjunct professor with the Huazhong University of Science and Technology, Wuhan, China. His recent research interests include: large-scale content distribution in the Internet, Peer-to-Peer media streaming, the Internet topology, cloud computing, green computing and communications. He is a fellow of the IEEE for “contribution to content distributions via the Internet”. He received the Young Investigator Award from the National Natural Science Foundation of China (NSFC) in 2004. He served as a Distinguished lecturer of the IEEE Communications Society (2006-2007). He was a corecipient for three Best Paper Awards from IEEE, and the Best System Track Paper in ACM Multimedia (2009).
\end{IEEEbiography}

\end{document}